\documentclass[aps,showpacs,preprintnumbers,amsmath,amssymb,twocolumn,superscriptaddress,floatfix,nofootinbib,10pt,pre]{revtex4-1}

\usepackage[lined,boxed,commentsnumbered,ruled,vlined]{algorithm2e}
\usepackage{algpseudocode}
\usepackage{amssymb, amsmath, amsfonts}
\usepackage{amsthm, mathrsfs, amsopn}
\usepackage{bm,bbm}
\usepackage{booktabs}
\usepackage{caption}
\usepackage{centernot}
\usepackage{colortbl}
\usepackage{comment}
\usepackage{dcolumn}
\usepackage{epstopdf}
\usepackage{float}
\usepackage{graphicx}
\usepackage[bottom]{footmisc}
\usepackage{mathrsfs}
\usepackage{mathtools}
\usepackage{multirow}
\usepackage{physics}
\usepackage{ragged2e}
\usepackage{subfigure}
\usepackage{tabularx}
\usepackage{tikz}
\usepackage{verbatim}
\usepackage{xkcdcolors}
\usepackage{hyperref}
\usepackage{cleveref}
\usepackage{amsthm}

\DeclareCaptionJustification{justified}{\justifying}
\hypersetup{colorlinks=true, citecolor=orange, urlcolor=blue, linkcolor=magenta}

\newtheorem{theorem}{Theorem}
\newtheorem{prop}[theorem]{Proposition}
\newtheorem{lemma}[theorem]{Lemma}
\newtheorem{remark}[theorem]{Remark}

\begin{document}

\title{Turing patterns in Matrix-Weighted Networks}

\author{Anna Gallo}
\affiliation{IMT School for Advanced Studies, Piazza San Francesco 19, 55100 Lucca (Italy)}
\affiliation{INdAM-GNAMPA Istituto Nazionale di Alta Matematica `Francesco Severi', P.le Aldo Moro 5, 00185 Rome (Italy)}

\author{Wilfried Segnou}
\affiliation{Department of mathematics and Namur Institute for Complex Systems, naXys\\University of Namur, Rue Graf\'e 2, B5000 Namur (Belgium)}

\author{Timoteo Carletti}
\email{timoteo.carletti@unamur.be}
\affiliation{Department of mathematics and Namur Institute for Complex Systems, naXys\\University of Namur, Rue Graf\'e 2, B5000 Namur (Belgium)}

\date{\today}

\begin{abstract}
Diffusion-driven instability is a fundamental mechanism underlying pattern formation in spatially extended systems. In almost all existing works, diffusion across the links of the underlying network is modeled through scalar weights, possibly complemented by cross-diffusion terms that are homogeneous across links. In this work, we investigate the emergence of Turing patterns on Matrix Weighted Networks (MWNs), a recently introduced framework in which each edge is associated with a matrix weight. Focusing on the class of coherent MWNs, we provide a novel characterization of coherence in terms of node-dependent orthonormal matrices, showing that link transformations can be written as relative rotations between nodes.
This representation allows us to deal with coherent MWNs of any size and to introduce an orthonormal change of variables capable to reduce diffusion on a coherent MWN to diffusion on a standard weighted network with scalar weights. Building on this, we extend the classical Turing instability analysis to MWNs and derive the conditions under which a homogeneous equilibrium of the local dynamics loses stability due to matrix-weighted diffusion. \textcolor{black}{Moving beyond the dimensional constraints of previous approaches,} our results show how network topology, scalar weights, and inter-node transformations jointly shape pattern formation, and provide a constructive framework to analyze and design Turing patterns on matrix-weighted and higher-order networked systems. 
\end{abstract}


\maketitle
\section{Introduction}
\label{sec:intro}
Beautiful and colorful patterns emerge spontaneously in Nature and in human-made devices; they are the result of the interaction among the many basic units constituting the system~\cite{anderson,Nicolis1977,pastorsatorrasvespignani2010}. Very often, they originate from local reactions responsible for the creation and destruction of units, coupled with (long-range) diffusion capable of moving them apart. In this framework sits the elegant theory developed by Alan Turing in the context of morphogenesis~\cite{Turing52} and later extended to study self-organization in more general systems: under suitable conditions, any tiny, spatially dependent perturbation applied to a stable homogeneous stationary equilibrium, will drive the latter to a new, possibly patchy, solution, i.e., the spatial Turing pattern~\cite{pastorsatorrasvespignani2010}. The diffusion mechanism being the destabilizing factor, the above process is also known in the literature as a diffusion-driven instability.

The interaction among the basic units constituting the system, can thus be modeled by means of reaction-diffusion equations that govern the deterministic evolution of the concentrations both in time and space, the latter being a regular substrate~\cite{Murray2001} or a discrete one, e.g., a complex network~\cite{NM2010}. 

Since the work by Nakao and Mikhailov~\cite{NM2010}, scholars have studied the conditions for the emergence of Turing patterns in diverse scenarios ranging from directed networks~\cite{Asllani1} to non-normal ones~\cite{top_resilience,jtb}, passing through multiplex networks~\cite{Asllani2014,KHDG,busiello_homogLturing}. Motivated by the increasing interest in higher-order networks~\cite{battiston2020networks,BOCCALETTI20231,5.0151265,natphys,MillanEtAl2025}, Turing patterns have been studied on hypergraphs~\cite{carletti2020dynamical}, simplicial complexes~\cite{gao2023turing}, higher-order structures~\cite{muologallo,DORCHAIN2024115730} and recently in the framework of topological signals~\cite{turing_topological,MUOLO2024114312} (the interested reader can consult~\cite{rspa.2024.0235} for a recent review about Turing patterns on networks and higher-order networks). \textcolor{black}{In addition, in recent years, the classical framework of diffusion-driven instability on complex networks has been significantly expanded through novel analyses. For instance, in~\cite{asllani2025pattern}, evoking pattern formation theory, the authors explain the emergence of chimera states on complex networks, while, more recently, the article~\cite{khan2026bifurcation} has highlighted how the selection of complex bifurcation patterns on networks is shaped by the interplay between multi-scale local kinetics and the underlying topological structure.}

In all the above mentioned works the ``connections'' of the underlying substrate, being links of a network of some higher-order counterpart, i.e., hyperedges or simplexes, carry scalar weights that multiply species densities to encode the role of diffusion coefficients. A relevant and different case is the one where cross-diffusion is at play~\cite{GAMBINO20121112,fanelli_cianci,Gao_2020,cnad052}, namely, each node, or spatial location, receives a linear combination of species densities from each incoming link. In the case the state of each node is described by a $d$-dimensional vector, this process can be realized by multiplying the latter by a suitable $d\times d$ matrix. It is thus natural to consider that in the diffusion process, species densities undergo a transformation. To the best of our knowledge, in the literature, one can only find results dealing with the same ``diffusion'' matrix for all links in the framework of cross-diffusion models. 

In this work, we make one step further and consider the emergence of Turing patterns in the case where each link could be endowed with a different transformation matrix, by exploiting the recently introduced framework of \textit{Matrix Weighted Networks}~\cite{tian2025matrix,gallo2025global} (MWN).
\textcolor{black}{Unlike higher-order network frameworks such as hypergraphs or simplicial complexes, which extend the network topology to model non-dyadic interactions among groups of nodes, MWNs connect nodes in pairs, but enrich the definition of their coupling, by replacing scalar strengths with matrices. It follows that such a framework naturally accounts for multidimensional couplings and provides a novel geometric perspective to explore multi-dimensional patterns, which are different from the ones that emerge in frameworks based on multi-node interactions.
}

Our results apply to the class of MWNs satisfying the coherence condition~\cite{tian2025matrix}. Roughly speaking, the latter implies that any signal propagating through multiple oriented paths in a MWN always returns to its starting point without distortion, i.e., the product of the matrices weights across an oriented cycle is the identity. In this work, we will show that the dynamical system should preserve the MWN coherence, determining thus a strong interaction between dynamics and structure~\cite{MillanEtAl2025}. 

Turing instability manifests once the diffusive-like coupling is capable to turn unstable an otherwise locally stable spatially homogeneous solution. The coherence of a MWN and its invariance by the dynamics allow to show the existence of a suitable change of variables, thanks to which one can prove the existence of such stable spatially homogeneous solution. By resorting to a linear stability analysis about the latter solution, we can determine conditions on the spectrum of the MWN Laplace matrix to ensure the emergence of Turing patterns. \textcolor{black}{It is worth noting that the linearized stability analysis developed herein shares a strong conceptual foundation with the celebrated \textit{Master Stability Function} (MSF) framework widely adopted in synchronization theory~\cite{pecora1998master}. Both approaches rely on decoupling the linearized problem around a reference state, via a projection on the eigenvectors of the network Laplace matrix, by obtaining in this way a $1$-parameter family of linear systems depending each one on a given eigenvalue. The novelty of our MWN formulation, with respect to the classical MSF, lies in the fact that it allows us to model and analyze systems where coordinate transformations are inherently heterogeneous and embedded directly onto the network edges. Our theoretical contribution shows that, when the coherence condition is satisfied, it is possible to disentangle these edge-specific transformations via a global change of variables $\mathcal{S}$, effectively mapping a structurally disordered matrix-coupled system back to a tractable scalar-like dispersion relation.}
\textcolor{black}{Finally, we emphasize that the key novelty of our approach lies in shifting the complexity directly onto the network topology. Indeed, the existing literature already includes several works dealing with generalizations of Turing patterns on networks, but, while they fundamentally rely on standard networks where the interaction between two nodes is defined by a single scalar weight, the proposed matrix-weighted framework naturally enables multi-dimensional cross-coupling and spatial coordinate rotations between nodes.}

\textcolor{black}{To sum up, the main theoretical contributions and innovative aspects of this paper are multifold. First, in Section~\ref{sec:char}, we overcome the dimensional constraints of previous approaches by introducing a novel characterization of coherent MWNs, by providing a framework applicable to networks of arbitrary size. Second, in Section~\ref{sec:TTMWN}, we prove that coherence condition allows us to disentangle edge-specific transformations via a global change of variables and map the multi-dimensional reaction-diffusion process onto a standard scalar-like dispersion relation. Finally, in Section~\ref{sec:results}, we show how to build coherent matrix networks from any standard network topology, and validate our theoretical framework by using numerical simulations of the Stuart-Landau model~\ref{ssec:SL}, of an abstract model with a given rotation invariance~\ref{ssec:abst}, and of the Lorenz model~\ref{ssec:Lorenz}. For a summary of our findings and an outline of future research directions, see Section~\ref{sec:conc}.}


\section{Characterization of Weighted Matrix Networks}
\label{sec:char}

The aim of this section is to briefly introduce {\em Matrix Weighted Networks} (MWNs)~\cite{tian2025matrix}, present a novel characterization of coherent MWNs, and introduce the main tools we need to develop a Turing theory for dynamical systems coupled via MWNs. As in~\cite{tian2025matrix}, we restrict our attention to reciprocal interactions. Accordingly, without loss of generality, we present the discussion in terms of symmetric (undirected) networks.\\

Let $G=(V,E,\{\mathbf{W}_{ij}\}_{i,j=1}^n)$ be a MWN with $n=|V|$ nodes and $m=|E|$ edges. To any existing edge, $(i,j)\in E$, we associate a weighted matrix $\mathbf{W}_{ij}\in \mathbb{R}^{d\times d}$ defined as
\begin{align}
\label{eq:matrix}
    \mathbf{W}_{ij} = w_{ij}\mathbf{R}_{ij}\, ,
\end{align}
where $w_{ij}:=||\mathbf{W}_{ij}||_2>0$ is the scalar weight and $\mathbf{R}_{ij}\in\mathbb{R}^{d\times d}$ is a transformation, such that $\|\mathbf{R}_{ij}\|_2=1$. In the following, we assume the latter to belong to the group of rotations and  $\mathbf{W}_{ij}=\mathbf{W}_{ji}^\top$ for all $i$ and $j$, meaning $w_{ij}=w_{ji}$ and $\mathbf{R}_{ij}=\mathbf{R}_{ji}^\top$. 

Let $d_i=\sum_j w_{ij}$ denote the strength of node $i$, and define the \textit{supra-degree matrix} $\mathcal{D}=\mathbf{D}\otimes \mathbf{I}_d$, where $\mathbf{D}=\mathrm{diag}(d_1,\dots,d_n)$, $\mathbf{I}_d$ is the $d$-dimensional identity matrix and $\otimes$ denotes the Kronecker product. We can then define the {\em supra-Laplace matrix} as
\begin{align}
\label{eq:supraL}
\mathcal{L}=\mathcal{D}-\mathcal{W}\, ,
\end{align}
where $\mathcal{W}$ is the \textit{supra-weight matrix}, which has a block structure, with the $(i,j)$-block being $\mathbf{W}_{ij}$.\\

\paragraph*{Coherence condition and characterization of coherent MWN.} A MWN is said to be \textit{coherent}~\cite{tian2025matrix} if, for every oriented cycle composed by $k$ different edges, $\mathcal{C}:=((i_1,i_2),(i_2,i_3),\dots,(i_k,i_1))$, the product of the transformation matrices along the cycle equals the identity, i.e.,
\begin{align}  
\label{eq:cohercond1}
\prod_{(i,j)\in\mathcal{C}}\mathbf{R}_{ij}=\mathbf{I}_d\,.
\end{align}
This property implies that nodes can be partitioned into distinct groups, with the transformation along any walk between nodes within the same group being the identity. Furthermore, the transformation from any node in one group to any node in another group is the same.

In a coherent MWN, one can define the block diagonal matrix $\mathcal{S}$, where the $i$--th block is the $d\times d$ matrix, $\mathbf{O}_{1i}$, representing the product of transformations, i.e., rotations in the present setting, along any oriented walk starting from node $1$ and ending at node $i$ (see Eq.~\eqref{eq:matS} below for further details). Since the composition of rotations is itself a rotation, each $\mathbf{O}_{1i}$ is also a rotation. Importantly, because of the coherence condition, the choice of the first node can be arbitrary (for the sake of convenience, we hereby label it as node $1$), and if several paths exist, the choice of the walk is also arbitrary and does not affect the results. Moreover, in the partition associated with the MWN, the matrix $\mathbf{O}_{1i}$ represents the composed transformation between any node in the group to which $1$ belongs and any node in the group of $i$. To check that a given MWN is coherent can be, in principle, computationally costly because it requires to verify~\eqref{eq:cohercond1} for all paths; for this reason the theory has been so far applied only to simple hand-made examples of coherent MWN~\cite{tian2025matrix,gallo2025global}. One goal of this work is to fill this gap and to propose an algorithm capable to create coherent MWN of any size, based on the following proposition.\\

\begin{prop}[Characterization of coherent MWNs]\label{prop:coherence}
Let $G$ be a MWN, whose topology is given by an oriented, symmetric network. The following are equivalent.
\begin{itemize}
    \item[(a)] There exist $n$ orthonormal matrices $\mathbf{Q}_1,\dots,\mathbf{Q}_n\in O(d)$ such that, for every edge $(i,j)\in E$,
    \begin{align}
        \mathbf{W}_{ij} = w_{ij}\mathbf{Q}_i^\top \mathbf{Q}_j\, ,
    \end{align}
    namely $\mathbf{R}_{ij}=\mathbf{Q}_i^\top \mathbf{Q}_j$.
    \item[(b)] For every $(i,j)\in E$, the link matrix $\mathbf{R}_{ij}$ is orthonormal and, for any oriented cycle 
    $\mathcal C=((v_0,v_1),(v_1,v_2),\dots,(v_{\ell-1},v_{\ell}))$ in $G$, with $v_{\ell}=v_0$,
    \begin{align}
        \prod_{s=1}^\ell \mathbf{R}_{v_{s-1}v_s} = \mathbf{I}_d,
    \end{align}
    i.e., the graph $G$ is coherent.
\end{itemize}
\end{prop}

\begin{proof}
\textit{(a) $\Rightarrow$ (b)}. Assume that there exist orthonormal matrices $\mathbf{Q}_1,\dots,\mathbf{Q}_n\in O(d)$ such that $\mathbf{W}_{ij} = w_{ij}\mathbf{Q}_i^\top \mathbf{Q}_j$ for all $(i,j)\in E$. Then, for any $(i,j)\in E$ we can define
\begin{align}
\label{eq:defQtQ}
    \mathbf{R}_{ij}=\mathbf{Q}_i^\top \mathbf{Q}_j,
\end{align}
which is orthonormal, being the product of orthonormal matrices. Now, consider any path $P:i=v_0,v_1,\dots,v_\ell=j$. The product of transformations along the path is
\begin{align}
    \prod_{s=1}^\ell \mathbf{R}_{v_{s-1}v_s} 
    &= \prod_{s=1}^\ell \mathbf{Q}_{v_{s-1}}^\top \mathbf{Q}_{v_s}\nonumber\\
    &= \mathbf{Q}_i^\top \left(\prod_{s=1}^{\ell-1} \mathbf{Q}_{v_s}\mathbf{Q}_{v_s}^\top\right)\mathbf{Q}_j
    = \mathbf{Q}_i^\top \mathbf{Q}_j\, .
\end{align}
If $i=j$, the path is a cycle, and the product equals $\mathbf{I}_d$, confirming the coherence condition holds.

\medskip
\textit{(b) $\Rightarrow$ (a)}.  
Assume that for every edge $(i,j)\in E$, $\mathbf{R}_{ij}$ is orthonormal, and, for every cycle $\mathcal C$, the product of the rotations along it is the identity $\mathbf{I}_d$.  
Fix a reference node $k\in V$ and select an arbitrary orthonormal matrix $\mathbf{Q}_k$ (e.g., $\mathbf{I}_d$).  
For each $j\in V$, consider a path $P_{kj}:k=v_0,v_1,\dots,v_\ell=j$, and define
\begin{align}
    \widetilde{\mathbf{Q}}_j := \mathbf{Q}_k \mathbf{R}_{P_{kj}}, \qquad 
    \mathbf{R}_{P_{kj}} := \prod_{s=1}^\ell \mathbf{R}_{v_{s-1}v_s}.
\end{align}
Notice that, if ${P'_{kj}}$ is another path from $k$ to $j$, concatenating ${P_{kj}}$ and ${P'_{kj}}$ forms a cycle. Then, by coherence,
\begin{align}
    \mathbf{R}_{P_{kj}} \mathbf{R}_{{P'_{kj}}}^{-1} = \mathbf{I}_d,
\end{align}
so $\mathbf{R}_{P_{kj}}=\mathbf{R}_{{P'_{kj}}}$ and $\widetilde{\mathbf{Q}}_j$ is well defined.  

By construction,
\begin{align}
    \widetilde{\mathbf{Q}}_i^\top \widetilde{\mathbf{Q}}_j
    = (\mathbf{Q}_k \mathbf{R}_{P_{ki}})^\top (\mathbf{Q}_k \mathbf{R}_{P_{kj}})
    = \mathbf{R}_{P_{ki}}^\top \mathbf{R}_{P_{kj}}.
\end{align}
Now, for any edge $(i,j) \in E$, consider the cycle formed by 
concatenating the path $P_{ki}$, the edge $(i,j)$, and the reverse 
path $P_{jk}$. By coherence, 
\begin{align}
    \mathbf{R}_{P_{ki}} \mathbf{R}_{ij} \mathbf{R}_{P_{jk}}^{-1} = \mathbf{I}_d.
\end{align}
Since $\mathbf{R}_{P_{jk}}$ is orthonormal, $\mathbf{R}_{P_{jk}}^{-1} = \mathbf{R}_{P_{jk}}^\top = \mathbf{R}_{P_{kj}}$. Therefore,
\begin{align}
    \mathbf{R}_{ij} = \mathbf{R}_{P_{ki}}^\top \mathbf{R}_{P_{kj}} = 
    \widetilde{\mathbf{Q}}_i^\top \widetilde{\mathbf{Q}}_j.
\end{align}
\end{proof}

\textcolor{black}{It is worth emphasizing that the result presented in Proposition~\ref{prop:coherence}, as well as the cycle-basis property established in Lemma~\ref{lemma:cohebasis}, strictly depend on the assumption of symmetric networks, where $\mathbf R_{ij}^\top = \mathbf R_{ji}$.}

Notice that the representation 
\begin{align}
\mathbf{R}_{ij} = \mathbf{Q}_i^\top \mathbf{Q}_j, \qquad (i,j)\in E,
\end{align}
is unique up to a global orthonormal transformation. 
Indeed, if $\{\mathbf{Q}_i\}_{i=1}^n$ satisfy condition \textit{(a)}, then for any 
$\mathbf{O}\in O(d)$ the matrices 
$\mathbf{Q}_i' = \mathbf{Q}_i \mathbf{O}$ 
define the same network, since
\begin{align}
\mathbf{Q}_i'^\top \mathbf{Q}_j' 
= \mathbf{O}^\top \mathbf{Q}_i^\top \mathbf{Q}_j \mathbf{O}
= \mathbf{Q}_i^\top \mathbf{Q}_j.
\end{align}
Thus, the family $\{\mathbf{Q}_i\}_{i=1}^n$ is determined only up to a right action of $O(d)$. For more details about the characterization of coherent MWNs, we refer the interested reader to Appendix~\ref{app:coherence}.\\


\paragraph*{Supra-Laplace matrix and Identity-Transformed MWNs.}

Let us introduce the block diagonal matrix $\mathcal{S}$, whose $i$--th block is the $d\times d$ matrix, $\mathbf{O}_{1i}$, above defined
in formula:
\begin{align}
\label{eq:matS}
\mathcal{S}=\left(
\begin{matrix}
 \mathbf{I}_d & 0 & \dots & \dots & 0\\
 0 & \mathbf{O}_{12} & \dots & \dots & 0\\
 \vdots & \vdots & \ddots & \vdots & \vdots\\
 0 & \dots & \dots & 0 & \mathbf{O}_{1n}
\end{matrix}\right)\, .
\end{align}
Let us observe that, in the following, without loss of generality, we will assume each block to contain a single node. On the contrary, if several nodes are contained in each block, the above formula is valid up to a small abuse of notation because we identified node $i$ with the $i$--th block to which it belongs; this will not change the results, because all nodes in the $i$--th block are equivalent, being connected by links whose transformation is the identity matrix. This is the general MWNs framework introduced in~\cite{tian2025matrix}.

Given any $\vec{u}\in\mathbb{R}^{d}$, the vector $\vec{U}=\mathcal{S}^\top(\vec{1}_n\otimes \vec{u})$ can be shown to be an eigenvector of the supra-Laplacian ${\mathcal{L}}$ with eigenvalue $0$ if and only if the MWN is coherent~\cite{tian2025matrix}. 

By using the matrix $\mathcal{S}$ and the supra-Laplacian $\mathcal{L}$, we can define a second supra-Laplace matrix 
\begin{align}
\bar{\mathcal{L}} = \mathcal{S}\mathcal{L}\mathcal{S}^\top\, ,
\end{align}
which corresponds to the supra-Laplace matrix where all the transformation matrices $\mathbf{O}$ 
have been replaced by the identity matrix $\mathbf{I}_d$, as described in~\cite{tian2025matrix}. This new matrix depends only on the scalar weights and the topology of the network; indeed, it can be written as:
\begin{align}
    \label{eq:Lbar}
    \bar{\mathcal{L}}=\bar{\mathbf{L}}\otimes \mathbf{I}_d\, ,
\end{align}
where $\bar{\mathbf{L}}=\mathbf{D}-\mathbf{A}^{(w)}$, and $\mathbf{A}^{(w)}$ is the weighted adjacency matrix of the underlying network, i.e., each $(i,j)$--entry of $\mathbf{A}^{(w)}$ equals $w_{ij}$.


\section{Turing theory on MWN\MakeLowercase{s}}
\label{sec:TTMWN}

Let us consider a $d$-dimensional system whose state variable $\vec{x}\in\mathbb{R}^d$ evolves in time according to the ordinary differential equation (ODE)
\begin{align}
    \label{eq:isolated}
    \frac{d\vec{x}}{dt}=\vec{f}(\vec{x})\,,
\end{align}
where $\vec{f}:A\rightarrow\mathbb{R}^d$ is some nonlinear function and $A$ an open subset of $\mathbb{R}^d$. Let us assume to have $n$ identical copies of the above ODE, each one being identified by the state variable $\vec{x}_j\in\mathbb{R}^d$, $j=1,\dots,n$, and assume moreover to couple them via a MWN. The time evolution of the state variable anchored to the $i$--th node, $\vec{x}_i$, is then described by
\begin{align}
\label{eq:dynsyscoupled2}
\frac{d\vec{x}_i}{dt}=& \vec{f}(\vec{x}_i)- \sum_j \mathcal{L}_{ij}\vec{h}(\vec{x}_j),\notag\\
=&\vec{f}(\vec{x}_i) -\sum_j {w}_{ij}\mathbf{I}_d \vec{h}(\vec{x}_j) + \sum_j \mathbf{W}_{ij}\vec{h}(\vec{x}_j)\, ,
\end{align}
where we hypothesized that the coupling function $\vec{h}$ does not depend on the node index. By defining the vector $\vec{x}=(\vec{x}_1^\top,\dots,\vec{x}_n^\top)^\top\in \mathbb{R}^{nd}$, we can rewrite Eq.~\eqref{eq:dynsyscoupled2} as
\begin{align}
\label{eq:dynsyscoupledLmw2}
\frac{d\vec{x}}{dt}={f}_*(\vec{x})- \mathcal{L}h_*(\vec{x})\, ,
\end{align}
where $f_*$ and $h_*$ act component-wise resulting into a $nd$-dimensional vector, namely
\begin{align}
 {f}_*(\vec{x}):=(\vec{f}(\vec{x}_1)^\top,\dots,\vec{f}(\vec{x}_n)^\top)^\top\, ,
\end{align}
and similarly for $h_*$.

A Turing instability occurs if  system~\eqref{eq:dynsyscoupledLmw2} admits a stable homogeneous stationary solution, $\vec{x}_j=\vec{x}^*$ for all $j=1,\dots,n$, once we silence the coupling, i.e., we set $\mathcal{L}=0$, that turns out unstable once the coupling is taken into account. The instability is therefore driven by the coupling and leads the system to a new equilibrium, possibly heterogeneous, i.e., spatially dependent, known as a Turing pattern.

Let us, thus, assume the ODE~\eqref{eq:isolated} possesses a stable stationary solution, $\vec{x}^*$, we would like the latter to be also a solution of the coupled system, namely $\vec{x}_j(t)\equiv\vec{x}^*$, for all $j=1,\dots, n$ and $t\geq 0$, to solve~\eqref{eq:dynsyscoupled2}. In the case the underlying support is a (connected) network with positive weight, then the above claim holds true~\cite{fujisaka1983stability,Pecora_etal97,pecora1998master} because the network Laplace matrix admits the eigenvector $\vec{1}_n=(1,\dots,1)^\top\in\mathbb{R}^n$ with eigenvalue $\Lambda^{(1)}=0$. For general MWNs the latter fact does not arise because the transformation matrices ``mix'' the components of the state vector; however one can ``disentangle'' those modes by showing that $\vec{X}^*=\mathcal{S}^\top(\vec{1}_n\otimes \vec{x}^*)$ is a stationary solution of~\eqref{eq:dynsyscoupledLmw2}. The latter claim holds true if the MWN is coherent~\cite{tian2025matrix,gallo2025global} and if the dynamical system~\eqref{eq:dynsyscoupledLmw2} preserves the coherence, namely it is invariant with respect to the (product of) matrices $\mathbf{R}_{ij}$~\cite{gallo2025global}. Namely, for all $i=1,\dots,n$ and all $\vec{x}\in\mathbb{R}^d$, the following conditions must hold true:
\begin{align}
\label{eq:condfO}
 \mathbf{O}_{1i}\vec{f}(\mathbf{O}_{1i}^\top \vec{x}) = \vec{f}(\vec{x})\quad\text{and}\quad \mathbf{O}_{1i}\vec{h}(\mathbf{O}_{1i}^\top \vec{x}) = \vec{h}(\vec{x}) \, .
\end{align}

The time derivative of $\vec{X}^*$ clearly vanishes, being the latter a constant vector. On the other hand, we have
\begin{eqnarray*}
 \vec{f}_*(\vec{X}^*)-\mathcal{L}\vec{h}_*(\vec{X}^*)&=&\left(\begin{smallmatrix}
     \vec{f}(\vec{x}^*)\\ \vec{f}(\mathbf{O}^\top_{12}\vec{x}^*)\\\vdots\\ \vec{f}(\mathbf{O}^\top_{1n}\vec{x}^*)
 \end{smallmatrix}\right)-\mathcal{L}\left(\begin{smallmatrix}
     \vec{h}(\vec{x}^*)\\ \vec{h}(\mathbf{O}^\top_{12}\vec{x}^*)\\\vdots\\ \vec{h}(\mathbf{O}^\top_{1n}\vec{x}^*)
\end{smallmatrix}\right)\\
&=&\left(\begin{smallmatrix}
     \vec{f}(\vec{x}^*)\\ \mathbf{O}^\top_{12}\vec{f}(\vec{x}^*)\\\vdots\\ \mathbf{O}^\top_{1n}\vec{f}(\vec{x}^*)
 \end{smallmatrix}\right)-\mathcal{L}\left(\begin{smallmatrix}
     \vec{h}(\vec{x}^*)\\ \mathbf{O}^\top_{12}\vec{h}(\vec{x}^*)\\\vdots\\ \mathbf{O}^\top_{1n}\vec{h}(\vec{x}^*)
 \end{smallmatrix}\right)\,,
\end{eqnarray*}
where we used the invariance of $\vec{f}$ and $\vec{h}$ given by~\eqref{eq:condfO}. We can thus conclude that
\begin{align}
\label{eq:equil}
 \vec{f}_*(\vec{X}^*)&-\mathcal{L}\vec{h}_*(\vec{X}^*)=\\
&=\mathcal{S}^\top (\vec{1}_n\otimes \vec{f}(\vec{x}^*))-\mathcal{L}\mathcal{S}^\top (\vec{1}_n\otimes \vec{h}(\vec{x}^*))=0\notag\, ,
\end{align}
because $\vec{f}(\vec{x}^*)=0$ and $\mathcal{L}\mathcal{S}^\top (\vec{1}_n\otimes \vec{h}(\vec{x}^*))=0$.

The emergence of Turing patterns relies on the proof of the instability of the stationary solution $\vec{X}^*$. To achieve this goal we rewrite the $(dn)$-dimensional state vector as $\vec{x}=\vec{X}^*+\delta\vec{x}$, where $\delta\vec{x}$ is ``small'' perturbation; if $\delta\vec{x}(t)$ will converge to $0$, then the equilibrium $\vec{X}^*$ is locally \textcolor{black}{asymptotically} stable, and unstable otherwise. The time evolution of $\delta\vec{x}$ can be obtained by performing a linear stability analysis of ~\eqref{eq:dynsyscoupledLmw2} about the stationary solution $\vec{X}^*$:
\begin{align}
\label{eq:lin1}
\frac{d\delta\vec{x}}{dt}=\mathbf{J}_{f_*}(\vec{X}^*)\delta\vec{x}-\mathcal{L}\mathbf{J}_{h_*}(\vec{X}^*)\delta\vec{x}\, ,
\end{align}
where $\mathbf{J}_{f_*}(\vec{X}^*)$ and $\mathbf{J}_{h_*}(\vec{X}^*)$ are respectively the Jacobian of ${f_*}$ and ${h_*}$ evaluated on the equilibrium $\vec{X}^*$. The latter equation can be written in ``components'' by defining $\delta\vec{x}=(\delta x_1^\top,\dots,\delta x_n^\top)^\top$ and obtain 
\begin{align}
\frac{d\delta \vec{x}_j}{dt}
 =\mathbf{J}_{f}(\mathbf{O}_{1j}^\top\vec{x}^*)
 \delta \vec{x}_j-\sum_\ell \mathcal{L}_{j\ell}\mathbf{J}_{h}(\mathbf{O}_{1\ell}^\top\vec{x}^*)
 \delta \vec{x}_\ell\, ,
\end{align}

The invariance condition~\eqref{eq:condfO} returns the following relations satisfied by the Jacobian matrices
\begin{align}
 \mathbf{J}_{f}(\mathbf{O}_{1j}^\top\vec{x}^*)=\mathbf{O}_{1j}^\top\mathbf{J}_{f}(\vec{x}^*)\mathbf{O}_{1j}
\end{align}
and
\begin{align}\mathbf{J}_{h}(\mathbf{O}_{1j}^\top\vec{x}^*)=\mathbf{O}_{1j}^\top\mathbf{J}_{h}(\vec{x}^*)\mathbf{O}_{1j}\, .
\end{align}
Hence, we can conclude that
\begin{align}
\frac{d\delta \vec{x}_j}{dt}
 =\mathbf{O}_{1j}^\top\mathbf{J}_{f}(\vec{x}^*)\mathbf{O}_{1j}
 \delta \vec{x}_j-\sum_\ell \mathcal{L}_{j\ell}\mathbf{O}_{1\ell}^\top\mathbf{J}_{h}(\vec{x}^*)
 \mathbf{O}_{1\ell}\delta \vec{x}_\ell\, ,
\end{align}
or equivalently by defining $\delta\vec{w}_j=\mathbf{O}_{1j}\delta\vec{x}_j$
\begin{eqnarray}
\label{eq:lin3}
\frac{d\delta \vec{w}_j}{dt}
 &=&\mathbf{J}_{f}(\vec{x}^*)
 \delta \vec{w}_j-\sum_\ell \mathbf{O}_{1j}\bar{L}_{j\ell}\mathbf{O}_{1\ell}^\top\mathbf{J}_{h}(\vec{x}^*)
 \delta \vec{w}_\ell\notag\\
 &=&\mathbf{J}_{f}(\vec{x}^*)
 \delta \vec{w}_j-\sum_\ell \bar{{L}}_{j\ell}\mathbf{J}_{h}(\vec{x}^*)
 \delta \vec{w}_\ell\, .
\end{eqnarray}
In the following, we will refer to $\delta\vec{w}_j$ as ``rotated'' variables because of the application of the rotation matrices $\mathbf{O}_{1j}$ onto the original variables, $\delta\vec{x}_j$. Let us observe that by using the stack vectors $\delta\vec{x}$ and $\delta\vec{w}$, the above change of variables can be rewritten as $\delta\vec{w}=\mathcal{S}\delta\vec{x}$. Let us finally observe that to get~\eqref{eq:lin3}, we made use of the definition of the supra-Laplace matrix $\bar{\mathcal{L}}$. 

To prove the stability of $\delta \vec{w}_j$ and hence of $\delta \vec{x}_j$, we exploit the existence of an orthonormal basis for the Laplace matrix $\bar{\mathbf{L}}$, i.e., $\bar{\phi}^{(\alpha)}$, $\Lambda^{(\alpha)}$, $\alpha=1,\dots,n$, to project $\delta \vec{w}_j$ onto the latter
\begin{align}
\delta \vec{w}_j=\sum_\alpha
\delta \hat{w}_\alpha\bar{\phi}^{(\alpha)}_j\, .
\end{align}
In this way, Eq.~\eqref{eq:lin3} returns
\begin{align}
\label{eq:lin4}
\sum_\alpha\frac{d\delta \hat{w}_\alpha}{dt}
\bar{\phi}^{(\alpha)}_j=\mathbf{J}_{f}(\vec{s})\sum_\alpha
\delta \hat{w}_\alpha\bar{\phi}^{(\alpha)}_j- \sum_\alpha \Lambda^{(\alpha)}\mathbf{J}_{h}(\vec{s})
\delta \hat{w}_\alpha\bar{\phi}^{(\alpha)}_j\, .
\end{align}
By left multiplying by $\bar{\phi}^{(\alpha)}$ and by using the orthonormality of eigenvectors we eventually obtain
\begin{eqnarray}
\frac{d\delta \hat{w}_\alpha}{dt}
 &=& \left[\mathbf{J}_{f}(\vec{x}^*)-\Lambda^{(\alpha)}\mathbf{J}_{h}(\vec{x}^*)\right]
\delta \hat{w}_\alpha \quad \forall \alpha=1,\dots,n\notag\\
&=:&\mathbf{J}_{\alpha}\delta \hat{w}_\alpha\, .
\end{eqnarray}

The above linear system contains the information about the dynamics and the coupling via the Jacobian matrices, while the MWN enters only via the eigenvalues of the supra-Laplace matrix $\bar{\mathcal{L}}$, depending only on the scalar weights. The perturbation $\delta \vec{w}_j$ does not converge to zero if there exists at least one $\alpha$ for which $\delta\hat{w}_\alpha$ does not converge to zero either. 

In conclusion, by defining the {\em dispersion relation}, $\lambda(\Lambda^{(\alpha)})$, to be the largest real part of the eigenvalues of the matrix $\mathbf{J}_{\alpha}$, then the existence of $\alpha$ such that $\lambda(\Lambda^{(\alpha)})>0$ determines the instability conditions we were looking for.

{\color{black}
In conclusion, the Turing instability threshold is exactly determined by the eigenvalues $\Lambda^{(\alpha)}$ of the Laplace matrix. It is known that by using spectral graph theory one can related the eigenvalues to the average degree; hence, one could link such a spectral dependency to the network structural properties. Specifically, for an Erd\H{o}s-R\'{e}nyi topology with connection probability $p$ and average degree $\langle k\rangle=np$, the bulk of the Laplacian spectrum can be approximated using the effective support range $[\langle k\rangle-2\sqrt{\langle k\rangle},\langle k\rangle+2\sqrt{\langle k\rangle}]$. This means that an increase of the average degree, $\langle k \rangle$, shifts the entire spectrum towards higher values leading to conclude that denser networks allow higher-order, highly-frequent spatial modes to cross the critical threshold $\Lambda_{\text{crit}}$, thereby directly altering the frequency of the emerging patterns. Let us however observe that the average degree of the underlying network could not be a determining factor for the onset of Turing pattern because the latter depends on the coherence property of the MWN. Stated differently, the same underlying network could support or not Turing pattern depending on the choice of the matrices $\mathbf{R}_{ij}$.}

In the remaining sections, we will present the above theory applied to three relevant dynamical systems, but its validity clearly goes beyond those examples.

\begin{table*}[ht]
\centering
\label{tab:symbols}
\begin{tabular}{lll}
\hline
\textbf{Symbol} & \textbf{Description} & \textbf{Dimension / Space} \\ \hline
$n$ & Number of nodes in the network & $\mathbb{N}$ \\
$\mathbf{R}_{ij}$ & Orthogonal transformation matrix & $SO(d)$ \\
$\mathbf{W}_{ij}$ & Matrix weight of the edge between nodes $i$ and $j$ & $\mathbb R^{n \times n}$ \\ 
$\mathcal S$ & Synchronization transformation operator & $\mathbb R^{dn\times dn}$ \\
$\mathcal L$ & Supra-Laplacian matrix of the MWN & $\mathbb R^{dn\times dn}$ \\
$\bar{\mathbf L}$ & Underlying scalar weighted Laplacian matrix & $\mathbb R^{n\times n}$ \\
$\bar{\mathcal L}$ & Identity-transformed supra-Laplacian matrix ($\bar{\mathbf L}\otimes\mathbf I_d$) & $\mathbb R^{nd\times nd}$  \\
$\Lambda^{(\alpha)}$ & Eigenvalues of the scalar network Laplacian $\bar{\mathbf L}$ & $\mathbb R$ \\
$\vec x_j$ & State vector of the $j$-th node (original variables) & $\mathbb R^d$ \\
$\vec w_j$ & Rotated state vector of the $j$-th node ($\mathbf O_{1j}\vec x_j$) & $\mathbb R^d$ \\
$\vec x^*$ & Stable stationary equilibrium of the isolated local dynamics & $\mathbb R^d$ \\
$\vec X^*$ & Transformed-synchronized global stationary solution $\mathcal S^\top(\vec{1}_n \otimes \vec{x}^*)$ & $\mathbb R^{nd}$ \\
$\mathbf{J}_{f}, \mathbf{J}_{h}$ & Jacobian matrices of the reaction and coupling functions & $\mathbb R^{d \times d}$ \\
$\mathbf M(\Lambda^{(\alpha)})$ & Mode Jacobian matrix & $\mathbb R^{d\times d}$ \\
$\lambda(\Lambda^{(\alpha)})$ & Dispersion relation (largest real part of eigenvalues of $\mathbf M$) & $\mathbb R$ \\
$\mathbf J_f, \mathbf J_h$ & Jacobian matrices of the reaction and coupling functions (Stuart-Landau) & $\mathbb R^{d \times d}$ \\
$z_j, w_j$ & Complex state and rotated variables (Stuart-Landau / Abstract model) & $\mathbb C$ \\
$\vec \xi_j = (\xi_j, \eta_j)^\top$ & Real and imaginary components of the rotated variables & $\mathbb R^2$ \\
$\mathbf E$ & Selection/mixing coupling matrix (Lorenz system) & $\mathbb R^{3\times 3}$ \\
\hline
\end{tabular}
\caption{\textcolor{black}{\textbf{Summary of key mathematical symbols and notations used throughout the manuscript.}}}
\end{table*}

\section{Results}
\label{sec:results}

The aim of this section is to introduce three systems and examine the conditions under which Turing instability can emerge as predicted by the theory developed above.

\subsection{The Stuart–Landau model}
\label{ssec:SL}

The first system we take into account is the Stuart-Landau (SL) model~\cite{Stuart1978,vanharten,aranson,garcamorales}, a canonical example of
nonlinear oscillators widely used to describe a broad class of phenomena and resulting to be a normal form for
systems close to a supercritical Hopf-bifurcation. For the application we want to describe, we will, however, consider the SL dynamics in the subcritical case, where, i.e., the origin is a stable equilibrium. For a detailed analysis and the explicit computations presented in this section, we refer to Appendix~\ref{app:SL}.

In Cartesian coordinates, a single SL oscillator $j$ can be written as
\begin{align}\label{eq:SL}
\frac{d}{dt}
\begin{pmatrix}
 x_j\\y_j
\end{pmatrix} =&
\begin{pmatrix}
 \sigma_{\Re} & -\sigma_{\Im}\\
 \sigma_{\Im} & \sigma_{\Re}
\end{pmatrix}
\begin{pmatrix}
 x_j\\y_j
\end{pmatrix}\nonumber\\
&-  (x_j^2+y_j^2)
\begin{pmatrix}
 \beta_{\Re} & -\beta_{\Im}\\
 \beta_{\Im} & \beta_{\Re}
\end{pmatrix}
\begin{pmatrix}
 x_j\\y_j
\end{pmatrix}\, ,
\end{align}
where we introduced the complex model parameters $\sigma=\sigma_{\Re}+i\sigma_{\Im}$ and $\beta=\beta_{\Re}+i\beta_{\Im}$. 

Here, we consider $n$ identical SL oscillators anchored to the nodes of a MWN, coupled
via a diffusive-like nonlinear function. The dynamics of the $j$--th unit is thus given by
\begin{widetext}
\begin{align}
\label{eq:SLnonlin}
\frac{d}{dt}
\begin{pmatrix}
 x_j\\y_j
\end{pmatrix} &=
\begin{pmatrix}
 \sigma_{\Re} & -\sigma_{\Im}\\
 \sigma_{\Im} & \sigma_{\Re}
\end{pmatrix}
\begin{pmatrix}
 x_j\\y_j
\end{pmatrix} -  (x_j^2+y_j^2)
\begin{pmatrix}
 \beta_{\Re} & -\beta_{\Im}\\
 \beta_{\Im} & \beta_{\Re}
\end{pmatrix}
\begin{pmatrix}
 x_j\\y_j
\end{pmatrix}-\sum_\ell \mathcal{L}_{j\ell}\left[ (x_\ell^2+y_\ell^2)^{\frac{m-1}{2}}
\begin{pmatrix}
 \mu_{\Re} & -\mu_{\Im}\\
 \mu_{\Im} & \mu_{\Re}
\end{pmatrix}
\begin{pmatrix}
 x_\ell\\y_\ell
\end{pmatrix}\right]\nonumber\\
&=:\vec{f}(x_j,y_j)-\sum_\ell \mathcal{L}_{j\ell}\vec{h}(x_\ell,y_\ell)\, ,
\end{align}
\end{widetext}
where $\vec f(x_j,y_j)$ is the nonlinear function defined by the above equation and $\vec{h}(x_\ell,y_\ell):=(x_\ell^2+y_\ell^2)^{\frac{m-1}{2}}
\begin{pmatrix}
 \mu_{\Re} & -\mu_{\Im}\\
 \mu_{\Im} & \mu_{\Re}
\end{pmatrix}
\begin{pmatrix}
 x_\ell\\y_\ell
\end{pmatrix}$ defines the coupling function with complex coupling strength $\mu=\mu_{\Re}+i\mu_{\Im}$, and $\mathcal L$ is the supra-Laplace matrix of the MWN.

We define the underlying MWN by following the construction provided in Proposition~\ref{prop:coherence} and focusing on the two-dimensional case ($d=2$). More specifically, to each node $i$ we associate an orthonormal matrix $\mathbf R_i \in \mathrm{O}(2)$ and, for every pair of connected nodes $(i,j)$, we construct the corresponding matrix-valued edge weights as $\mathbf W_{ij}=w_{ij}\mathbf R_{ij}$ where $w_{ij}\in\mathbb R$ is a scalar positive weight and the link transformation is defined by $\mathbf R_{ij} = \mathbf R_i^\top\mathbf R_j$. Note that $\mathbf R_{ij}$ is orthonormal as the product of orthonormal matrices.

We can associate with the complex parameters $\sigma$, $\beta$ and $\mu$ three real antisymmetric $2\times 2$ matrices, and one can prove that they commute with any $2\times 2$ orthonormal matrix $\mathbf R$. It thus follows
\begin{align}
    \vec f(\mathbf R\vec x)=\mathbf R\vec f(\vec x),\quad\vec h(\mathbf R\vec x)=\mathbf R\vec h(\vec x),\:\forall\vec x,
\end{align}
and hence the dynamics preserve the coherent structure of the network.

The origin is clearly an equilibrium of~\eqref{eq:SL} and one can prove that it is stable provided $\sigma_{\Re}<0$.

To determine the conditions for the onset of Turing instability, we linearize the coupled system~\eqref{eq:SLnonlin} in the original coordinates $(x_j,y_j)$ about the origin, then we perform the ``rotation'' to new coordinates, $(\xi_j,\eta_j)$, i.e., $(\xi_1,\eta_1,\dots,\xi_n,\eta_n)^\top=\mathcal{S}(x_1,y_1,\dots,x_n,y_n)^\top$, and project the resulting system onto the Laplace eigenbasis
\begin{align}
    \xi_j&=\sum_\alpha\hat{\xi}_\alpha\vec{\phi}^{(\alpha)}_j, \\\eta_j&=\sum_\alpha\hat{\eta}_\alpha\vec{\phi}^{(\alpha)}_j\, .
\end{align}
By exploiting the orthogonality of the eigenbasis, the resulting systems can be rewritten as 
\begin{eqnarray}
    \frac{d}{dt}\begin{pmatrix}
        \hat{\xi}_\alpha \\ \hat{\eta}_\alpha
\end{pmatrix} &=& \left[\begin{pmatrix}
 \sigma_{\Re} & -\sigma_{\Im}\\
 \sigma_{\Im} & \sigma_{\Re}
\end{pmatrix}-\Lambda^{(\alpha)}\begin{pmatrix}
 \mu_{\Re} & -\mu_{\Im}\\
 \mu_{\Im} & \mu_{\Re}
\end{pmatrix}\right]\begin{pmatrix}
        \hat{\xi}_\alpha \\ \hat{\eta}_\alpha
\end{pmatrix}\notag\\
&=:&\mathbf{M}(\Lambda^{(\alpha)})\begin{pmatrix}
        \hat{\xi}_\alpha \\ \hat{\eta}_\alpha
\end{pmatrix}\, .
\end{eqnarray}
The diffusion-driven instability occurs when the real part of at least one eigenvalue of $\mathbf{M}(\Lambda^{(\alpha)})$ becomes positive for some $\alpha>1$. 

A direct computation allows to determine the  eigenvalues of $\mathbf M(\Lambda^{(\alpha)})$, to be $\lambda_\alpha = (\sigma_{\Re}-\Lambda^{(\alpha)}\mu_{\Re})
\pm i(\sigma_{\Im}-\Lambda^{(\alpha)}\mu_{\Im})$ for any $\alpha$.
Since $\sigma_{\Re}<0$ and $\Lambda^{(\alpha)} \geq 0$, we thus have
\begin{equation}
    \label{eq:condTSL}
\Re(\lambda_\alpha)=   \sigma_{\Re}-\Lambda^{(\alpha)}\mu_{\Re}>0\, ,
\end{equation}
if $\mu_{\Re}$ is \textit{sufficiently} negative.
Alternatively, once the model parameters $\sigma_{\Re}$ and $\mu_{\Re}$ are fixed, the instability condition~\eqref{eq:condTSL} holds true provided
\begin{align}\label{eq:condcritic}
\Lambda^{(\alpha)} > \frac{\sigma_{\Re}}{\mu_{\Re}} =: \Lambda_{\text{crit}},
\quad (\text{with }\alpha>1)\, .
\end{align}

Notice that such a condition shows that modes associated to eigenvalues $\Lambda^{(\alpha)} > \Lambda_{\text{crit}}$ 
become unstable, leading to pattern formation, while the uniform mode ($\alpha=1$, $\Lambda^{(1)}=0$) 
remains stable since $\sigma_{\Re}<0$. Finally, by analyzing the radial dynamics of the SL oscillator, we can show that no stable limit cycle exists for $\sigma_{\Re}<0$. Hence, the observed instability cannot be attributed to oscillatory behavior but is purely driven by the network coupling.\\



In Fig.~\ref{fig:SL_SBM} and Fig.~\ref{fig:SL_ER}, we provide numerical results supporting the analytical ones. In the former, the MWN is defined by using as underlying network obtained from a stochastic block model with \textcolor{black}{$n=500$} nodes divided into \textcolor{black}{$K=3$ blocks, the probability for a link to exist among any couple of nodes in the same block is $p_{in}=0.08$, while the probability to have a link among two nodes belonging to different blocks is $p_{out}=0.001$}. In the latter, the underlying topology of the MWN is generated by an Erd\H{o}s--R\'enyi random graph \textcolor{black}{composed by $n=500$ nodes and with a probability $p=0.02$ for each couple of nodes to be connected}. In both cases, top panels refer to the case where the dispersion relation is negative for all $\alpha$, thus the perturbation shrinks to zero and the origin is stable, while bottom panels show the onset of diffusion-driven instabilities due to the presence of unstable eigenvalues, i.e., those that satisfy condition~\eqref{eq:condcritic}. Panels $(a)$ and $(b)$ of Fig.~\ref{fig:SL_SBM} and Fig~\ref{fig:SL_ER} clearly show a linear dispersion relation, note that the blue curve has been drawn only to help the reader to identify the linear trend. 
The instability threshold is crossed precisely when $\Lambda^{(\alpha)} = \Lambda_{\text{crit}} = -\sigma_{\Re}/\mu_{\Re}$, and this determines the emergence of patterns as clearly visible in panels $(d)$ where we display the time evolution of $\xi_j(t)$, i.e., the real part of the complex signal $z_j(t)$ in the ``rotated'' variables.
\begin{figure*}
    \centering
    \includegraphics[width=1\linewidth]{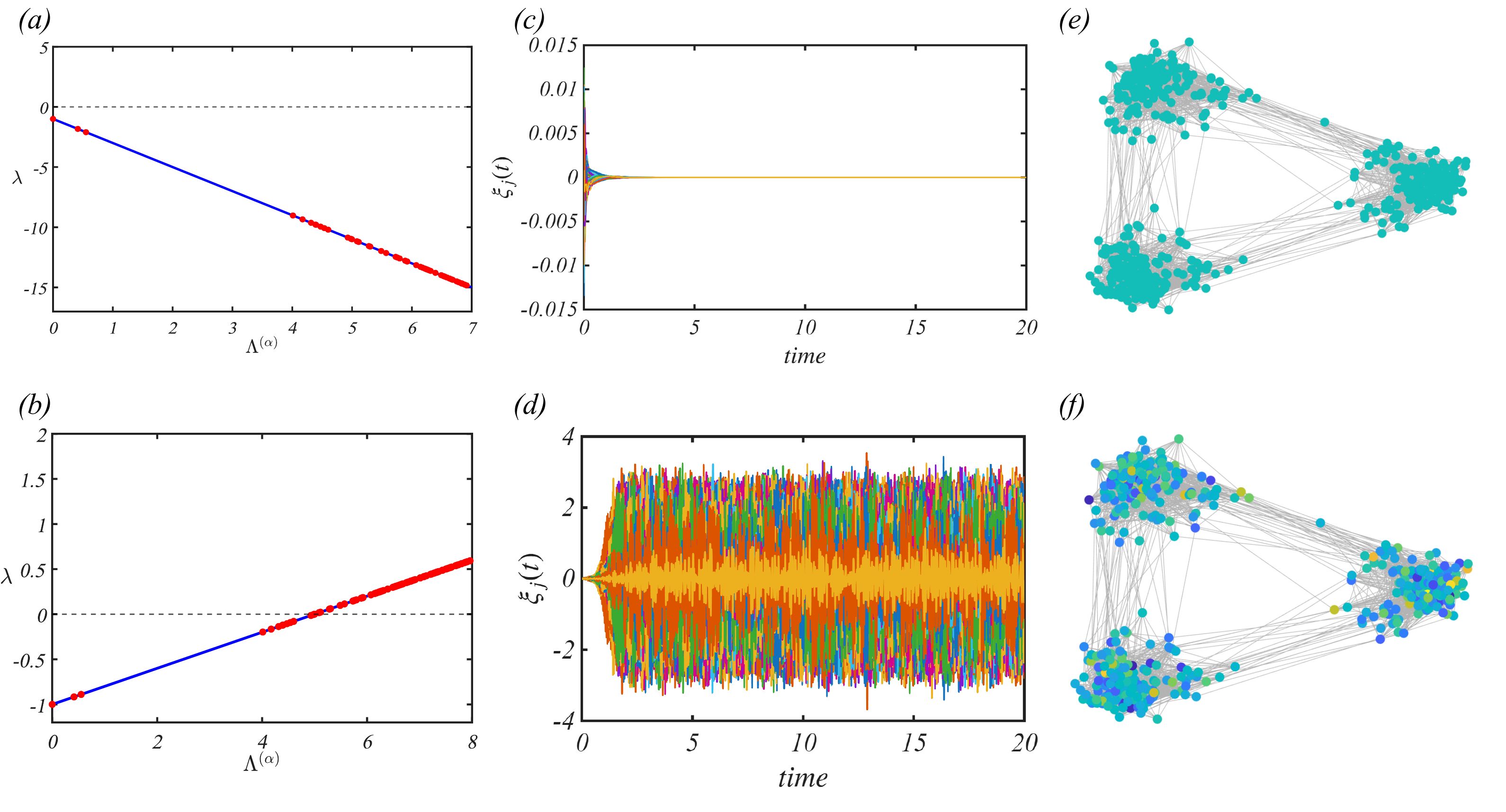}
    \caption{\textbf{Emergence of Turing patterns in a MWN with stochastic block model topology of coupled Stuart-Landau systems.} Panels $(a),(b)$ - The dispersion relation is reported as a function of the network Laplacian eigenvalues $\Lambda^{(\alpha)}$ (red dots), the blue curve is represented to emphasize the linear dependence: $(a)$ stable regime $\lambda(\Lambda^{(\alpha)})<0$ for all $\alpha$; $(b)$ unstable regime, there exist $\Lambda^{(\alpha)}$ associated to a positive dispersion relation, returning thus Turing pattern formation. Panels $(c),(d)$ - Temporal evolution of $\xi_j(t)$ across nodes: $(c)$ convergence to the homogeneous equilibrium of the oscillators in the stable regime; $(d)$ Turing patterns emerge in the unstable regime. Panels $(e),(f)$ - Network visualizations with node colors indicating dynamical states, values of $\xi_j(t)$ after a sufficiently long time: $(e)$ nodes present the same color, meaning that oscillators assume the same value; $(f)$ nodes present different colors indicating that in the unstable regime, Turing patterns emerge, i.e., $\xi_j$ vary from node to node. The model parameters are $\sigma = -1-0.5i$, $\beta = 1+i$, $m=1$,
    $\mu=2 + 5.5i$ for the top panels $(a)$, $(c)$, and $(e)$, while \textcolor{black}{$\mu=-0.2+5.5i$} for the bottom panels $(b)$, $(d)$, and $(f)$. The underlying topology is given by a  stochastic block model of Erd\H{o}s–Rényi networks composed by \textcolor{black}{$n = 500$ nodes and with $p_{in} = 0.08$ and $p_{out}=0.001$, and $K=3$ blocks.}}
    \label{fig:SL_SBM}
\end{figure*}

\begin{figure*}
    \centering
    \includegraphics[width=1\linewidth]{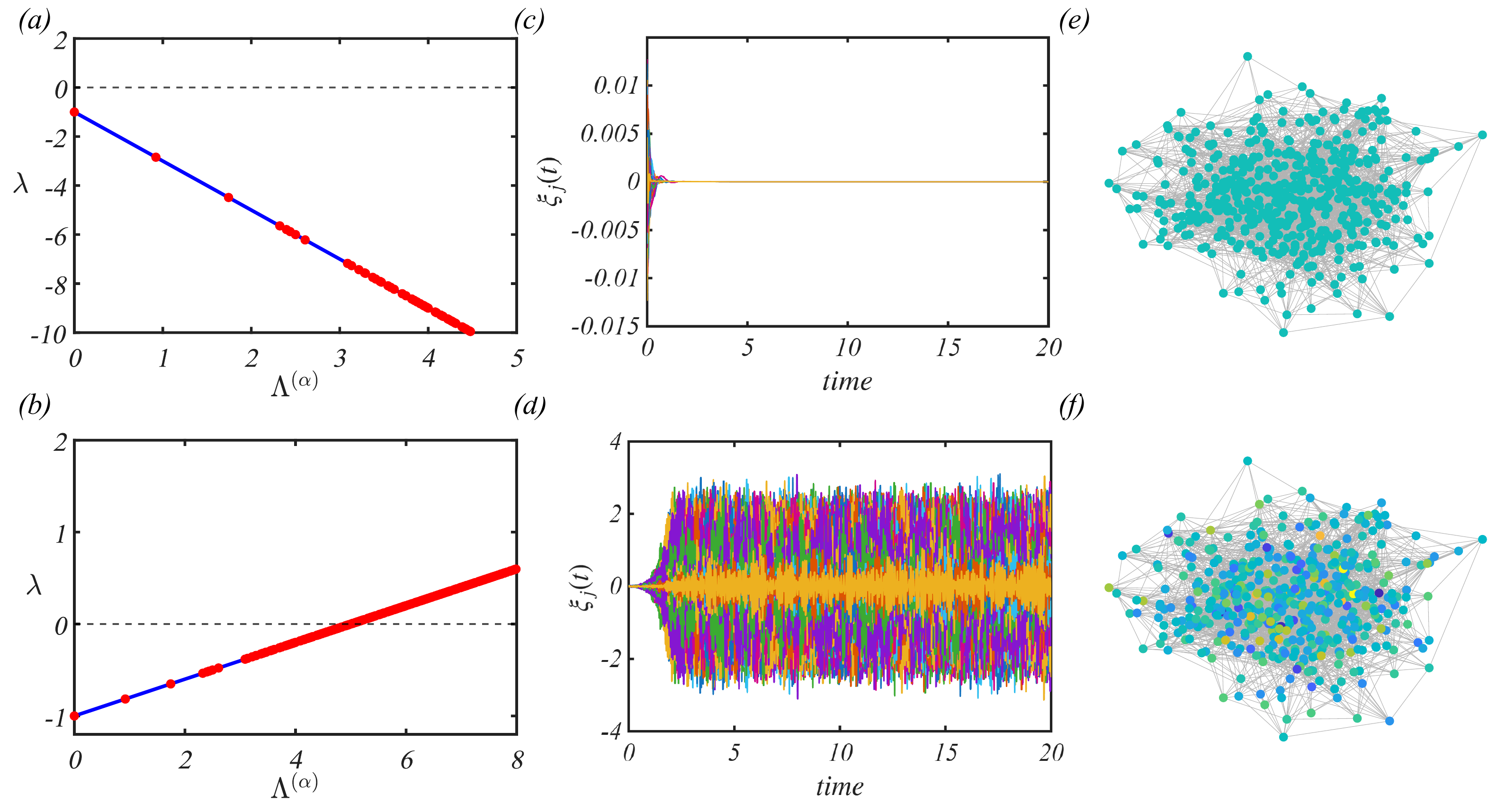}
    \caption{\textbf{Emergence of Turing patterns in a MWN with random Erd\H{o}s–Rényi topology of coupled Stuart-Landau systems.} Panels $(a),(b)$ - Dispersion relations is shown as a function of the network Laplacian eigenvalues $\Lambda^{(\alpha)}$ (red dots), the blue curve is displayed to emphasize the linear dependence: $(a)$ stable regime with a negative dispersion relation for all eigenvalues; $(b)$ unstable regime with dispersion relation assuming positive values for some eigenvalues, Turing pattern can thus develop. Panels $(c),(d)$ - Temporal evolution of $\xi_j(t)$ across nodes: $(c)$ convergence to the homogeneous solution in the stable regime; $(d)$ Turing patterns emerge in the unstable regime. Panels $(e),(f)$ - Network visualizations with node colors indicating dynamical states, i.e., value of $\xi_j(t)$ after a sufficiently long time: $(e)$ nodes present the same color, meaning that oscillators assumed the same value independently from the node index; $(f)$ nodes present different colors indicating that in the unstable regime, Turing patterns emerge, i.e., nodes differentiate among themselves. The model parameters are $\sigma = -1-0.5i$, $\beta = 1+i$, $\mu=2 + 5.5i$, $m=1$ for the top panels $(a)$, $(c)$, and $(e)$, while \textcolor{black}{$\mu=-0.2+5.5i$} for the bottom panels $(b)$, $(d)$, and $(f)$.  The underlying topology is given by a Erd\H{o}s–Rényi network composed by \textcolor{black}{$n=500$ nodes and $p=0.02$}}
    \label{fig:SL_ER}
\end{figure*}

\subsection{Abstract model invariant under rotations by $2\pi/k$}
\label{ssec:abst}

Let us, now, consider an abstract model with a given rotational symmetry to be used to test the emergence of Turing patterns in MWNs. The model is inspired by the Stuart-Landau system, where the cubic nonlinearity has been replaced with a general $(k+1)$--th power term, $k\geq 2$. More precisely, we consider
\begin{align}
\label{eq:modelk}
\frac{dz_j}{dt}=\sigma z_j+\beta z_j^{k+1}-\varepsilon \sum_\ell \mathcal{L}_{j\ell}z_\ell\, ,
\end{align}
where $z_j\in\mathbb{C}$ describes the state of the $j$--th node, $j=1,\dots,n$; $\sigma$, $\beta$ and $\varepsilon$ are complex parameters, $k\in\mathbb{N}$ determines the nonlinearity of the reaction part and $\mathcal{L}$ is the supra-Laplace matrix of the MWN. 

Once we silence the interaction via the MWN, we obtain for all $j=1,\dots,n$ the system
\begin{align}
\label{eq:modelkiso}
\frac{dz_j}{dt}=\sigma z_j+\beta z_j^{k+1}\equiv f(z_j)\, ,
\end{align}
that admits the trivial equilibrium $z^{(0)}=0$ and the $k$ roots of the equation $z^{k}=-\sigma /\beta$
\begin{align}
\label{eq:roots}
\hat{z}^{(s)} = \left(\frac{|\sigma|}{|\beta|}\right)^{\frac{1}{k}}e^{i\frac{\arg \sigma - \arg \beta}{k}}e^{i\pi \frac{1+2s}{k}}\quad \forall s=1,\dots k\, ,
\end{align}
where we introduced $\sigma=|\sigma|e^{i\arg \sigma}$ and similarly for $\beta$. The stability of such equilibria can be determined by linearizing the system about the equilibrium, namely, to compute the derivative of $f(z)$ at the equilibrium we are interested in, and to impose its real part to be negative:
\begin{align}
\label{eq:fprime}
f'(z^{(s)})=\sigma +(k+1)\beta (z^{(s)})^{k}=-\sigma k\, ,
\end{align}
where we used the definition of $z^{(s)}$, hence
\begin{align}
\label{eq:refprime}
\Re \left[f'(z^{(s)})\right]= -k\sigma_{\Re}\, ,
\end{align}
and being $k\geq 3$, stability occurs if $\sigma_{\Re}>0$, condition that we hereby assume to hold true.

Let us observe that Eq.~\eqref{eq:modelkiso} is invariant by rotation of a angle $2\pi/k$, indeed if we replace $z_j$ by $z_j e^{i 2\pi/k}$ then we obtain
\begin{align}
\label{eq:modelkisoinv}
e^{-i 2\pi/k}\frac{dz_j}{dt}=\sigma e^{-i 2\pi/k} z_j +\beta  e^{-i 2\pi (k+1)/k}(z_j)^{k+1}\, ,
\end{align}
from which we can conclude
\begin{align}
\label{eq:modelkisoinv2}
\frac{dz_j}{dt}=\sigma z_j +\beta z_j^{k+1}\, ,
\end{align}
namely the original system.

Fix $s\in\{1,\dots,k\}$. We are now interested in studying the stability of the solution $z_j=z^{(s)}$ for all $j=1,\dots,n$, for the coupled system~\eqref{eq:modelk}. We then introduce $z_j=z^{(s)}+u_j$, where $u_j\in\mathbb{C}$ is a small perturbation, and we perform a first-order expansion of Eq.~\eqref{eq:modelk}:

\begin{align}
\frac{du_j}{dt}&=\sigma u_j+\beta (k+1) (z^{(s)})^{k}u_j-\varepsilon \sum_\ell \mathcal{L}_{j\ell}u_\ell\nonumber\\
&= -k\sigma u_j-\varepsilon \sum_\ell \mathcal{L}_{j\ell}u_\ell\, .
\label{eq:modelklin}
\end{align}
By resorting again to the ``rotated'' variables, $\vec{w}=\mathcal{S}\vec{u}$, and by projecting them onto the supra-Laplace eigenbasis, we obtain:
\begin{align}
\label{eq:modelklinproj}
\frac{d\hat{w}_\alpha}{dt}= \left[-k\sigma-\varepsilon \Lambda^{(\alpha)}\right]\hat{w}_\alpha\, ;
\end{align}
the equilibrium solution $z_j=z^{(s)}$ is unstable if there exists $\alpha \geq 2$ such that 
\begin{align}
\label{eq:reldisp}
-k\sigma_{\Re}-\varepsilon_{\Re} \Lambda^{(\alpha)}>0\, ,
\end{align}
namely 
\begin{align}
\label{eq:reldisp2}
\varepsilon_{\Re} < -k\frac{\sigma_{\Re}}{\Lambda^{(\alpha)}}\, .
\end{align}

In Fig.~\ref{fig:AbstractMWN}, we report two cases supporting the analytical findings. \textcolor{black}{For the sake of clarity, note that, to ground the visual representation of our numerical patterns, we introduce the variable $\zeta_j(t)\in \mathbb{C}$, which corresponds to the components of the transformed vector $\vec{w}$ in the complex domain. Specifically, its real part, $\Re(\zeta_j(t))$, captures the dynamic evolution of the system by separating the dynamical changes from the static network rotations.} In the former one, the coupling parameter $\varepsilon$ does not satisfy condition~\eqref{eq:reldisp2} for any $\alpha$ and thus patterns cannot emerge (see top panels). On the other hand, in the second case, there exist several $\alpha$ such that condition~\eqref{eq:reldisp2} holds true for the chosen $\varepsilon$, and Turing patterns can emerge (see bottom panels). In both cases, the MWN is built by using a Barab\'ari-Albert network composed by \textcolor{black}{$n=500$} nodes where at each step a single link and a single node are added. The MWN is coherent because the matrix weights have been built according to Proposition~\ref{prop:coherence}, where the matrices $\mathbf{Q}_j$ are rotations by $2\pi/k$ with probability $q$ or the identity matrix with probability $1-q$; for this example we chosen $q=1/2$. 
\textcolor{black}{To conclude, let us briefly motivate such a choice. Note that the extreme cases $q=0$ and $q=1$ lead to structurally uniform configurations: in the first case all node-dependent matrices $\mathbf{Q}_j$ coincide with the identity $\mathbf{I}_d$, while in the second one each node shares the exact same rotation matrix $\mathbf{Q}_j = \mathbf{R}_{2\pi/k}$, implying that for any connected pair $(i,j)$ the link transformation yields $\mathbf{R}_{ij} = \mathbf{Q}_i^\top \mathbf{Q}_j = \mathbf{R}_{2\pi/k}^\top \mathbf{R}_{2\pi/k} = \mathbf{I}_d$, effectively canceling out any directional effect along the edges. It follows that the choice $q=1/2$ leads to maximize the structural disorder and the spatial mixing of coordinates across the network.}

\textcolor{black}{In addition, it is worth noting that, due to the coherence assumption, the system is invariant under changes in $q$. Indeed, since the transformation $\mathcal S$ is exact, the spectrum of the supra-Laplacian remains identical to that of the baseline scalar Laplacian (scaled by the identity block) and, therefore, no phase transition should occur as $q$ varies from $0$ to $1$, and the instability threshold remains constant.}

\begin{figure*}
    \centering
    \includegraphics[width=1\linewidth]{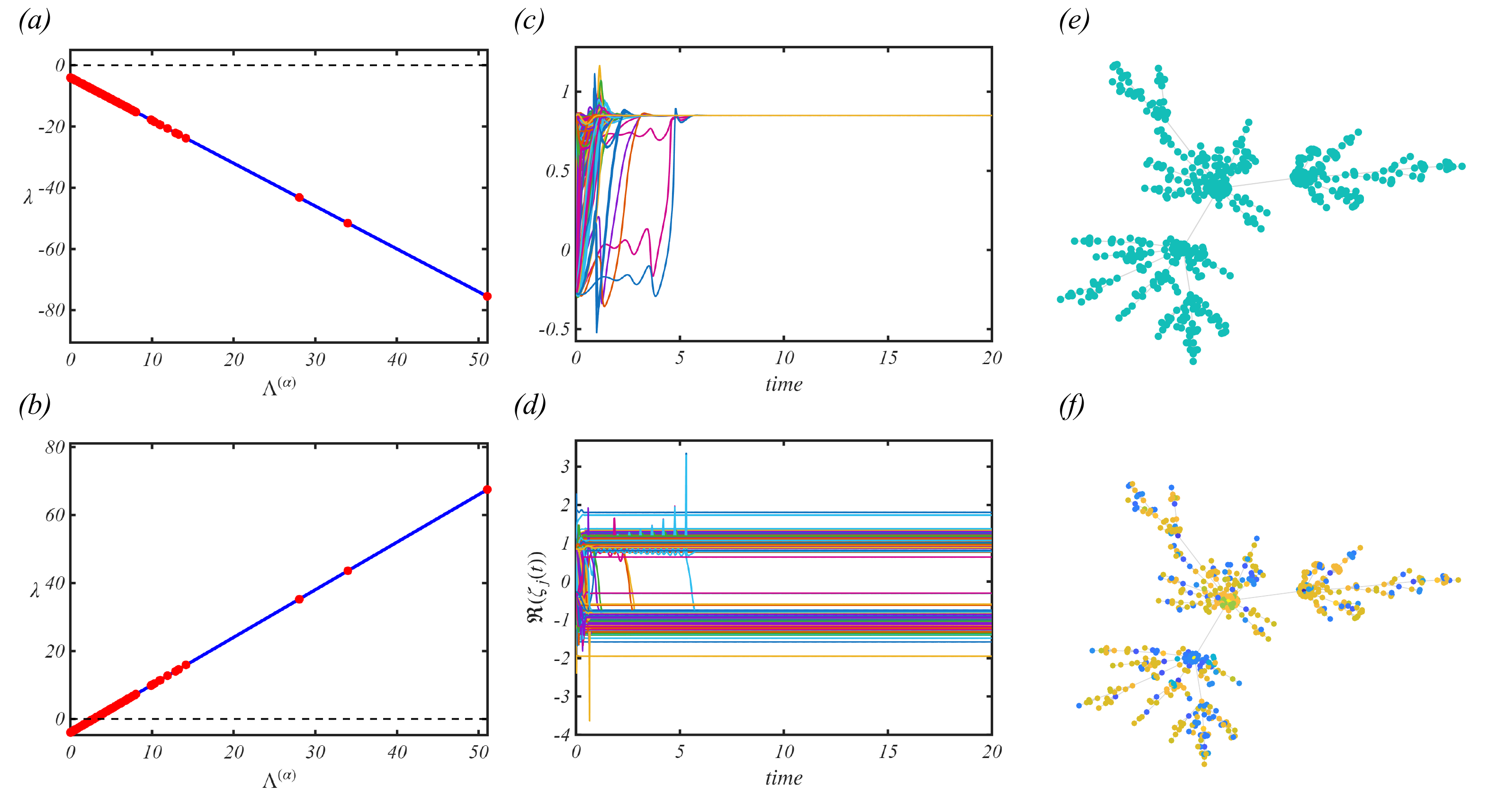}
    \caption{\textbf{Emergence of Turing patterns for the abstract model~\eqref{eq:modelk} defined on top of a MWN whose underlying topology is given by a Barab\'asi-Albert network composed by \textcolor{black}{$n=500$} nodes.} Panels $(a),(b)$ - Dispersion relations as a function of the Laplace eigenvalues $\Lambda^{(\alpha)}$ (red dots), the blue curve has been drawn to help the reader and can be obtained by replacing the eigenvalues with a continuous variable. In panel $(a)$, the dispersion relation remains negative for any value of $\Lambda^{(\alpha)}$, which prevents the emergence of Turing patterns, as can be observed in panel $(c)$, where we report the time evolution of $\Re(\zeta_j(t))$, and in panel $(e)$ where we show the network whose nodes are colored according to the asymptotic stationary values of $\Re(\zeta_j(t))$. In panel $(b)$, we report the dispersion relations and we can observe that it assumes positive values (red dots) for some $\Lambda^{(\alpha)}\gtrsim 3$, Turing patterns can thus emerge as visible in panel $(d)$, where we report $\Re(\zeta_j(t))$ versus time and panels $(f)$, where the network is displayed with the nodes again colored according to the asymptotic stationary values of $\Re(\zeta_j(t))$. The model parameters are $k=4$, $\sigma = 1+i$, $\beta = 1-2i$, $\varepsilon=1.4+2i$ for the top panels $(a)$, $(c)$, and $(d)$, while $\varepsilon=-1.4+2i$ for the bottom panels $(b)$, $(d)$, and $(f)$.}
    \label{fig:AbstractMWN}
\end{figure*}

Finally, let us observe that the results of Fig.~\ref{fig:AbstractMWN} have been obtained by using the variables ``rotated'' by the matrix $\mathcal{S}$, i.e., $\xi_j=\Re(w_j)$, and this choice is mandatory; indeed by looking at the results in the original variables can be misleading: one can observe a heterogeneous solution even once the dispersion relation is negative (see Fig.~\ref{fig:AbstractMWNwrongvar}), this is because of the mixing property induced by the transformations $\mathbf{R}_{ij}$ that is removed once we use the rotated variables.

\begin{figure*}
    \centering
\includegraphics[width=0.6\linewidth]{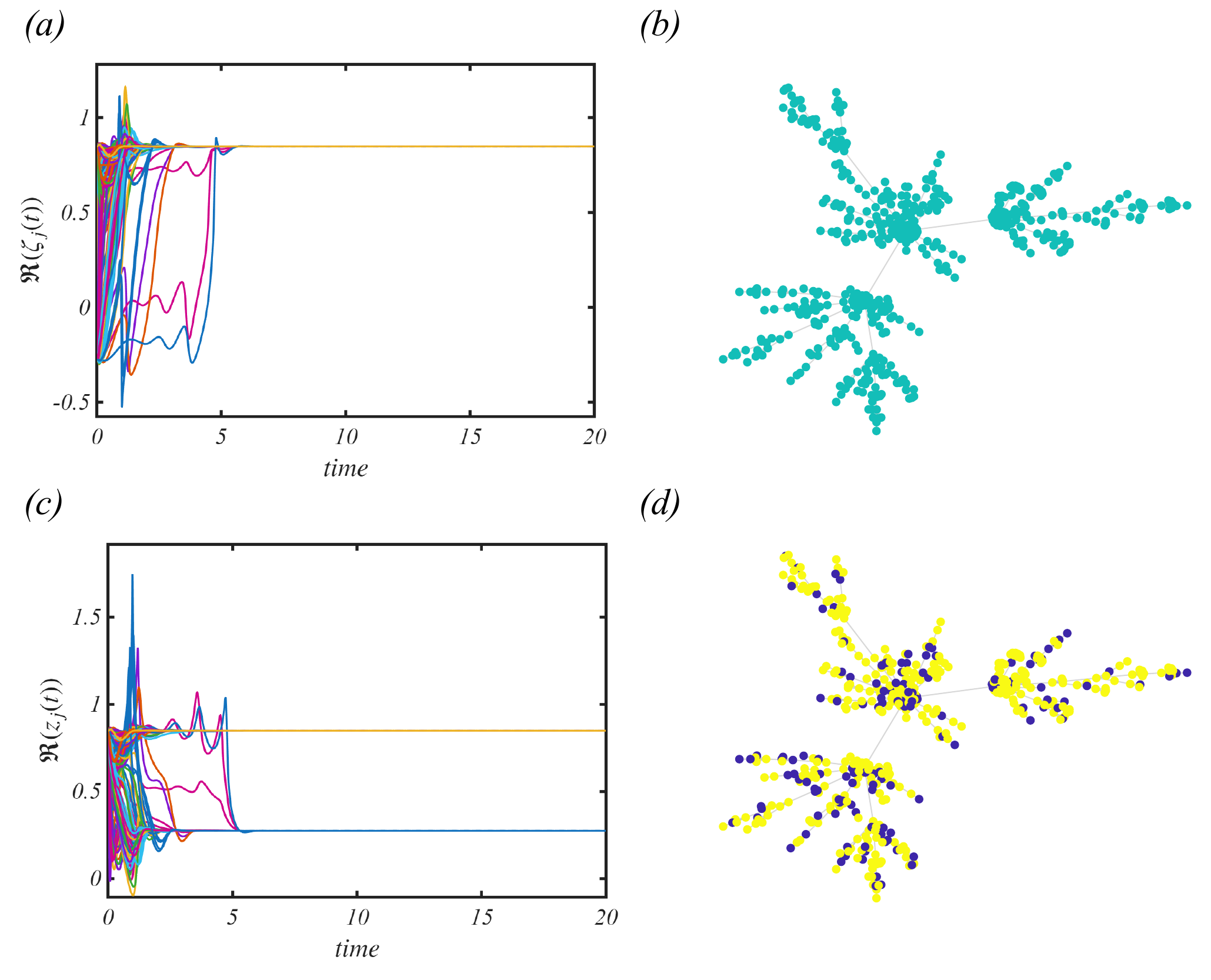}
    \caption{\textbf{On the use of the rotated variables and the original ones.} We consider the  model~\eqref{eq:modelk} with $k=4$ defined on top of a MWN whose underlying topology is given by a Barab\'asi-Albert network composed of \textcolor{black}{$n=500$} nodes. The parameters have been fixed so as to have a negative dispersion relation, $\sigma = 1+i$, $\beta = 1-2i$, and $\varepsilon=1.4+2i$. Hence, patterns cannot emerge as clearly shown in the top panels: (a) the time evolution of $\Re(\zeta_j(t))$, and (b) the network with nodes colored according to the asymptotic values of $\Re(\zeta_j(t))$. In the bottom panels, we show the same results but by using the original variables, (c) the time evolution of $\Re(z_j(t))$, and (d) the network with nodes colored according to the asymptotic values of $\Re(z_j(t))$. 
    \textcolor{black}{When observed in the original variables, the system shows an apparent spatial heterogeneity even before any instability occurs. This misleading impression of pattern formation is purely a geometric artifact, and such a phenomenon arises because the matrix weights introduce localized coordinate rotations along the network edges. Consequently, even when the system lies in a perfectly uniform steady state, these spatial rotations mix the local state variables differently at each node. This geometric artifact disappears once the proper global change of variables through $\mathcal S$ is applied.}}
    \label{fig:AbstractMWNwrongvar}
\end{figure*}

\subsection{The Lorenz model}
\label{ssec:Lorenz}

In the previous sections, we have considered one-dimensional complex systems (two-dimensional ones once we introduce real variables). The aim of this section is to provide an example in three dimensions. We hence consider a system of Lorenz oscillators, which provides a canonical example of a three-dimensional nonlinear dynamical system exhibiting rich bifurcation scenarios. Unlike the two-dimensional Stuart-Landau system, the Lorenz system requires a three-dimensional state space and possesses multiple equilibria. 

We now consider $n$ copies of the Lorenz system coupled through a MWN, described by 
\begin{align}
    \label{eq:lorenz_diffusive}
    \frac{d\vec{x}_j}{dt}=\vec{f}(\vec{x}_j)-\varepsilon \sum_\ell \mathcal{L}_{j\ell}\mathbf{E}\vec{x}_\ell\, ,
\end{align}
where $\vec x_j=(x_j,y_j,z_j)^\top$ represents the state of oscillator $j$, $\vec f$ is the Lorenz vector field, namely
\begin{align}
\label{eq:Lorenzf}
    \vec{f}(\vec{x}_j)=\left(\begin{matrix}
        \sigma(y_j-x_j)\\x_j(\rho-z_j)-y_j\\x_jy_j-\beta z_j
    \end{matrix}\right),
\end{align}
$\varepsilon>0$ is the coupling strength, $\mathcal L$ is the supra-Laplace matrix, and $\mathbf E\in\mathbb R^{3\times 3}$ describes how the dynamical variables are coupled. Note that according to Turing's theory, the equilibrium under consideration must be stable. For $\rho>1$, the uncoupled Lorenz system admits three equilibria: the origin and two symmetric non-trivial states
\begin{align}\label{eq:eqLorenz}
\vec{x}_{\pm}^*
    =
    \begin{pmatrix}
        \pm\sqrt{\beta(\rho-1)} \\
        \pm\sqrt{\beta(\rho-1)} \\
        \rho-1
    \end{pmatrix}\, .
\end{align}
In the following, we will consider $\vec{x}^*=\vec{x}^*_+$,
that results stable if $\sigma > \beta +1$ and $1<\rho<\rho_H := \frac{\sigma(\sigma + \beta + 3)}{\sigma - \beta - 1}$.

Following Proposition~\ref{prop:coherence}, we construct a coherent MWN, in such a way that the three-dimensional Lorenz system is invariant with respect to the link transformations, i.e., the rotations. To achieve this goal, we associate with each node $i=1,\dots,n$ an orthogonal matrix $\mathbf R_i\in O(3)$. By exploiting the symmetry properties of the Lorenz dynamics, we choose these matrices to have a block-diagonal structure: a $2\times2$ orthogonal block acting on the $(x,y)$--subspace and a unit entry in the $(3,3)$ position.

Each matrix $\mathbf{R}_i$ is randomly selected among two matrices which leave the function $\vec{f}$ unchanged, namely, the identity matrix $\mathbf{I}_3$ and the reflection matrix $\mathbf{R}_{\pi} =\left(\begin{smallmatrix}
    -1 & 0 &0\\0 & -1 & 0\\ 0 & 0 & 1
\end{smallmatrix}\right)$, both occurring with equal probability $q=1/2$. The key property is that $\mathbf R_\pi\vec f(\mathbf R_\pi\vec x)=\vec f(\vec x)$, $\forall\:\vec x$, ensuring invariance of the dynamics. For every pair of connected nodes $(i,j)$, we then construct matrix-valued edge weights so that the interaction between nodes $i$ and $j$ is weighted by the relative orthogonal transformation $\mathbf{R}_{ij} = \mathbf{R}_i^\top \mathbf{R}_j$.
These matrix-valued weights are embedded into a supra-adjacency matrix $\mathcal{W} \in \mathbb{R}^{nd \times nd}$, where each nonzero block corresponds to the matrix $\mathbf{R}_{ij}$, the amplitudes $w_{ij}$, all reduced to unity for simplicity's sake. 


Unlike the Stuart-Landau case, for the Lorenz system, the emergence of Turing instability relies on the study of the roots of a third-order polynomial. We could have used the explicit formula for the latter; however, we preferred to resort to the Routh-Hurwitz criterion (see Remark~\ref{rem:criterion}) and to numerical computations to analyze the dispersion relation.

Let us linearize Eq.~\eqref{eq:lorenz_diffusive} around its heterogeneous solution $\vec{X}^*=\mathcal{S}^\top(\vec{1}_n\otimes \vec{x}^*)$, where $\vec{x}^*=\vec{x}^*_+$ is given in Eq.~\eqref{eq:eqLorenz}, and by denoting with $\delta \vec{x}_j = \vec{x}_j - \vec{X}^*$ the perturbation of node $j$. By performing a first-order expansion, we get
\begin{align}
\frac{d\delta \vec{x}_j}{dt} = \mathbf{J}_f (\mathbf{O}_{1j}^\top \vec{x}^*) \delta \vec{x}_j - \varepsilon \sum_\ell \mathcal{L}_{j\ell} \mathbf{J}_h (\mathbf{O}_{1\ell}^\top \vec{x}^*) \delta \vec{x}_\ell 
\end{align}
where $\delta \vec{x} = (\delta x^\top, \delta y^\top, \delta z^\top)^\top$ and $\mathbf{J}_f (\vec{x}^*), \mathbf{J}_h (\vec{x}^*)$ are respectively the Jacobian of $\vec{f}$ and $\vec{h}(\vec{x})\equiv \mathbf{E}\vec{x}$ evaluated on the solution $\vec{x}^*$. By exploiting the invariance of the functions $\vec{f}$ and $\mathbf{E}\vec{x}$, we can introduce the new ``rotated'' variables $\delta \vec{w}_j = \mathbf{O}_{1j} \delta \vec{x}_j$ and obtain
\begin{align}
 \label{eq:linearequationMWnetwork}   
\frac{d\delta \vec{w}_j}{dt}  = \mathbf{J}_f (\vec{x}^*) \delta \vec{w}_j - \varepsilon\sum_\ell \bar{L}_{j\ell} \mathbf{E} \delta \vec{w}_\ell ,
\end{align}

Let us now project the perturbation $\delta \vec{w}_j$ onto the eigenbasis of the Laplace matrix $\bar{\mathbf{L}}$ i.e., $\delta \vec{w}_j = \sum_\alpha\hat{w}_\alpha\phi_j^{(\alpha)}$.  By inserting the latter into \eqref{eq:linearequationMWnetwork}, one gets 
\begin{align}
    \frac{d{\hat{w}}_\alpha}{dt} = \left(\mathbf J_f - \varepsilon\Lambda^{(\alpha)}\mathbf E\right)\hat{w}_\alpha,
\end{align}
where 
\begin{align}
    \mathbf J_f =
    \begin{pmatrix}
        -\sigma & \sigma & 0 \\
        \rho-z^* & -1 & -x^* \\
        y^* & x^* & -\beta
    \end{pmatrix}\, .
\end{align}


It follows that the stability of the equilibrium $\vec{X}^*$ can therefore be deduced by computing the eigenvalues of the matrices
\begin{align}
\label{eq:MlambdaLor}
\mathbf M(\Lambda^{(\alpha)}) =
\mathbf J_f - \varepsilon \Lambda^{(\alpha)} \mathbf{E}\, .
\end{align}
In particular, a diffusion-driven instability occurs when the real part of at least one eigenvalue of $\mathbf M(\Lambda^{(\alpha)})$ becomes positive for some $\alpha>1$. The characteristic equation associated with the latter matrix results to be a third-order polynomial that moreover depends on several model parameters. So instead of using the explicit formula for the third-order roots, we preferred to use a somewhat weaker result, but sufficient for our goal, which allows us to determine the sign of the root, namely the Routh-Hurwitz criterion.
\begin{remark}[Routh--Hurwitz criterion for cubic polynomials]\label{rem:criterion}
Consider a monic cubic polynomial
\begin{align}
p(\lambda)=\lambda^3+a_2\lambda^2+a_1\lambda+a_0,
\qquad a_i\in\mathbb{R}.
\end{align}
All roots of $p$ have strictly negative real part if and only if the following
Routh--Hurwitz conditions are satisfied:
\begin{align}
\label{eq:RHcubic}
a_2>0,\qquad a_1>0,\qquad a_0>0,\qquad a_2a_1>a_0.
\end{align}
\end{remark}

As a first application, we will hereby show that if the coupling matrix $\mathbf{E}$ has all zero elements but the diagonal ones, then Turing patterns cannot emerge. Indeed, in this case, we have
\begin{align}
\mathbf M(\Lambda^{(\alpha)}) =
\mathbf J_f - \varepsilon \Lambda^{(\alpha)} \left(\begin{matrix}
    e_1 & 0 & 0\\
    0 & e_2 & 0\\
    0 & 0 & e_3
\end{matrix}\right)\, ,
\end{align}
with $e_j\in\{0,1\}$. Hence the eigenvalues $\mu_j(\Lambda^{(\alpha)})$, $j=1,2,3$, of $\mathbf{M}(\Lambda^{(\alpha)})$ are given by
\begin{equation}
    \label{eq:eigMdiagE}
\mu_j(\Lambda^{(\alpha)}) = \mu_j(0) -\varepsilon \Lambda^{(\alpha)} e_j\, ,
\end{equation}
where $\mu_j(0)$ are the eigenvalues of $\mathbf{J}_f$. By assumption, the equilibrium $\vec{x}^*$ is stable, thus the latter has a negative real part and so do $\mu_j(\Lambda^{(\alpha)})$ for all $\alpha$ and $j=1,2,3$.\\

In Fig.~\ref{fig:lorenzER} we report numerical results supporting this finding. We consider a set of \textcolor{black}{$n=500$} Lorenz systems coupled through a random Erd\H{o}s-Rényi MWN, with a probability of occurrence of a link between two nodes $i$ and $j$ equal to $p = 0.15$. The model parameters have been set to $\sigma = 13$, $\beta = 8$, and $\rho = 28$, ensuring the stability of the equilibrium point $\vec{x}^*$ once the coupling is silenced. Top panels of Fig.~\ref{fig:lorenzER} refer to the linear coupling ${E}_{11} = 1$ and $\varepsilon=4$. As predicted by the analysis presented above, the dispersion relation is always negative (see panel $(a)$) and thus Turing instability does not hold, i.e., the equilibrium $\vec{x}_j=\vec{x}^*$ for all $j=1,\dots,n$, remains stable also in presence of the coupling (see panel $(c)$, where we plot the variable $\xi_j(t)$, i.e., the first component of the vector $\delta\vec{w}_j$) and state variables converge to the same value for all the nodes (see panel $(e)$). 
\begin{figure*}
    \centering
    \includegraphics[width=1\linewidth]{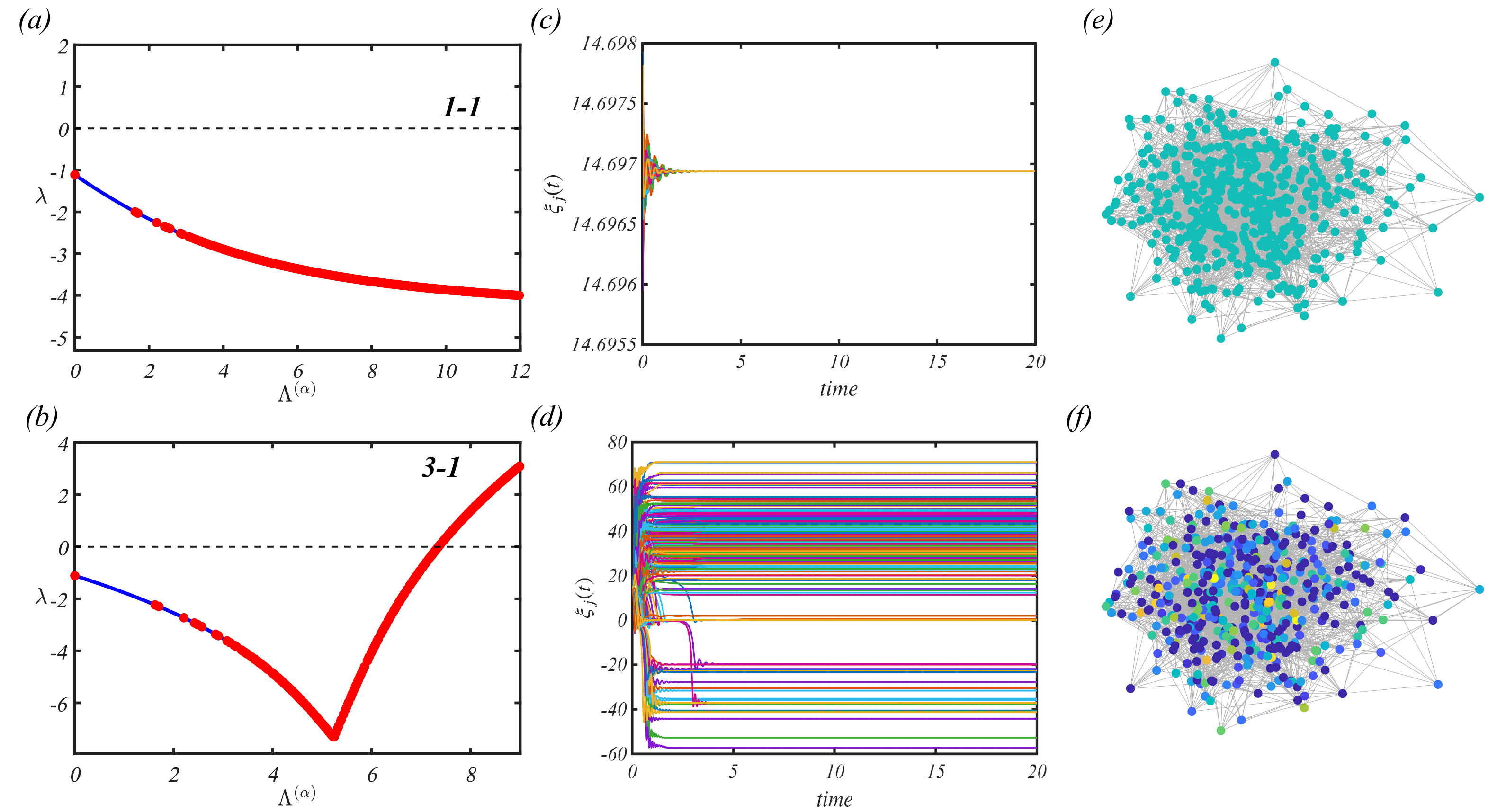}
    \caption{\textbf{Emergence of Turing patterns in a MWN with random Erd\H{o}s–Rényi topology of coupled Lorenz systems.} Panels $(a),(b)$ - The dispersion relations is displayed as a function of network Laplacian eigenvalue $\Lambda^{(\alpha)}$ (reds dots), the blue curve have been obtained by replacing the latter with a continuous variable and it is shown to help the reader to appreciate the nonlinear behavior: $(a)$ the dispersion relation remains negative for any value of $\Lambda^{(\alpha)}$, which prevents the emergence of Turing patterns; $(b)$ negative dispersion relationship up to a critical threshold of $\Lambda^{(\alpha)}$, beyond which it becomes positive enabling the emergence of Turing patterns. Panels $(c),(d)$ - Temporal evolution of $\xi_j(t)$ across nodes: $(c)$ all variables $\xi_j(t)$ converge to the equilibrium point $\xi^*$; $(d)$ emergence of Turing patterns. Panels $(e),(f)$ - Network visualizations with node colors indicating dynamical states, i.e., $\xi_j(t)$ after a sufficiently long period of time: $(e)$ nodes present same color meaning that oscillators reached the same value regardless of the node index; $(f)$ nodes present different colors indicating the onset of Turing patterns. The model parameters used to obtain the presented results are $\sigma = 13$, $\rho = 28$, $\beta=8$.  The underlying topology is given by a Erd\H{o}s–Rényi network composed by \textcolor{black}{$n = 500$ nodes and $p=0.02.$}, and the coupling is obtained with $\varepsilon=4$ and $E_{11}=1$ (top panels) while $E_{31}=1$ (bottom panels).}
    \label{fig:lorenzER}
\end{figure*}

Let us, now, consider a second case where the coupling mixes two different variables; for the sake of definiteness, we here consider ${E}_{31}=1$ and the remaining entries of $\mathbf{E}$ do vanish (the interested reader could find in Appendix~\ref{app:lorenz} the analysis for all the remaining cases). The characteristic polynomial of $\mathbf{M}(\Lambda^{(\alpha)})$ Eq.~\eqref{eq:MlambdaLor}, has coefficients
\begin{align}
\begin{aligned}
a_2 &= \sigma+\beta+1,\\
a_1 &= \beta(\rho+\sigma),\\
a_0 &= 2\sigma\beta(\rho-1)-\gamma\sigma\sqrt{\beta(\rho-1)},
\end{aligned}
\end{align}
where we introduced $\gamma=\varepsilon \Lambda^{(\alpha)}$ to lighten the notations. The first two conditions~\eqref{eq:RHcubic} are trivially satisfied because of the positivity of the parameters. The necessary conditions for the emergence of Turing patterns are thus $a_0<0$ or $a_2a_1-a_0<0$. The former one is equivalent to
\begin{equation}
    \label{eq:RHa3}
    2\beta(\rho-1)<\gamma\sqrt{\beta(\rho-1)}\, ,
\end{equation}
that requires (see Appendix~\ref{app:lorenz} for the detailed computation of $\gamma^{(crit)}_{31}$)
\begin{equation}
    \label{eq:RHa3bis}
    \varepsilon \Lambda^{(\alpha)} > 2\sqrt{\beta(\rho-1)}=\gamma^{(crit)}_{31}\, .
\end{equation}
Hence, if $\rho\ge 1$, one can obtain $a_0<0$ if $\Lambda^{(\alpha)}$ is sufficiently large. 

Let us now consider the remaining case $a_2a_1-a_0<0$. A straightforward computation returns
\begin{equation}
    \gamma\sigma\sqrt{\beta(\rho-1)}-2\beta\sigma(\rho-1)+\beta(\rho+\sigma)(\beta+\sigma+1)<0 \, ,
\end{equation}
which enables us to obtain
\begin{equation}
    \label{eq:cond2Lor}
    \varepsilon\Lambda^{(\alpha)}<2\sqrt{\beta(\rho-1)}-\frac{\beta(\rho+\sigma)(\beta+\sigma+1)}{\sigma\sqrt{\beta(\rho-1)}}=\gamma_2\, .
\end{equation}
Let us however observe (see Appendix~\ref{app:lorenz}) that $\gamma_2<0$ and thus condition~\eqref{eq:cond2Lor} is never verified, being $\Lambda^{(\alpha)}\geq 0$. In conclusion Turing instability arises if condition~\eqref{eq:RHa3bis} holds true.


In the bottom panels of Fig.~\ref{fig:lorenzER}, we show numerical results confirming the analytical ones. We consider again \textcolor{black}{$n=500$} Lorenz systems coupled through a random Erd\H{o}s-Rényi MWN, with a probability of occurrence of a link between any two nodes $i$ and $j$ equal to $p = 0.15$. The homogeneous equilibrium point $\vec{x}^*$ is stable, being $\varepsilon=4$, $\sigma = 13$, $\beta = 8$, and $\rho = 28$. The coupling is realized by assuming, $E_{31}=1$ and the remaining entries $E_{ij}=0$. One can observe the existence of large enough eigenvalues, $\Lambda^{(\alpha)}\gtrsim 7.35=\gamma^{(crit)}_{31}/\varepsilon$, for which the dispersion relation is positive (see red dots in panel $(b)$) and thus patterns do emerge (see panels $(d)$ and $(f)$ where we report, respectively, the time evolution of $\xi_j(t)$ and the value of the same variable after a sufficiently long time period).
\begin{figure*}
    \centering
    \includegraphics[width=1\linewidth]{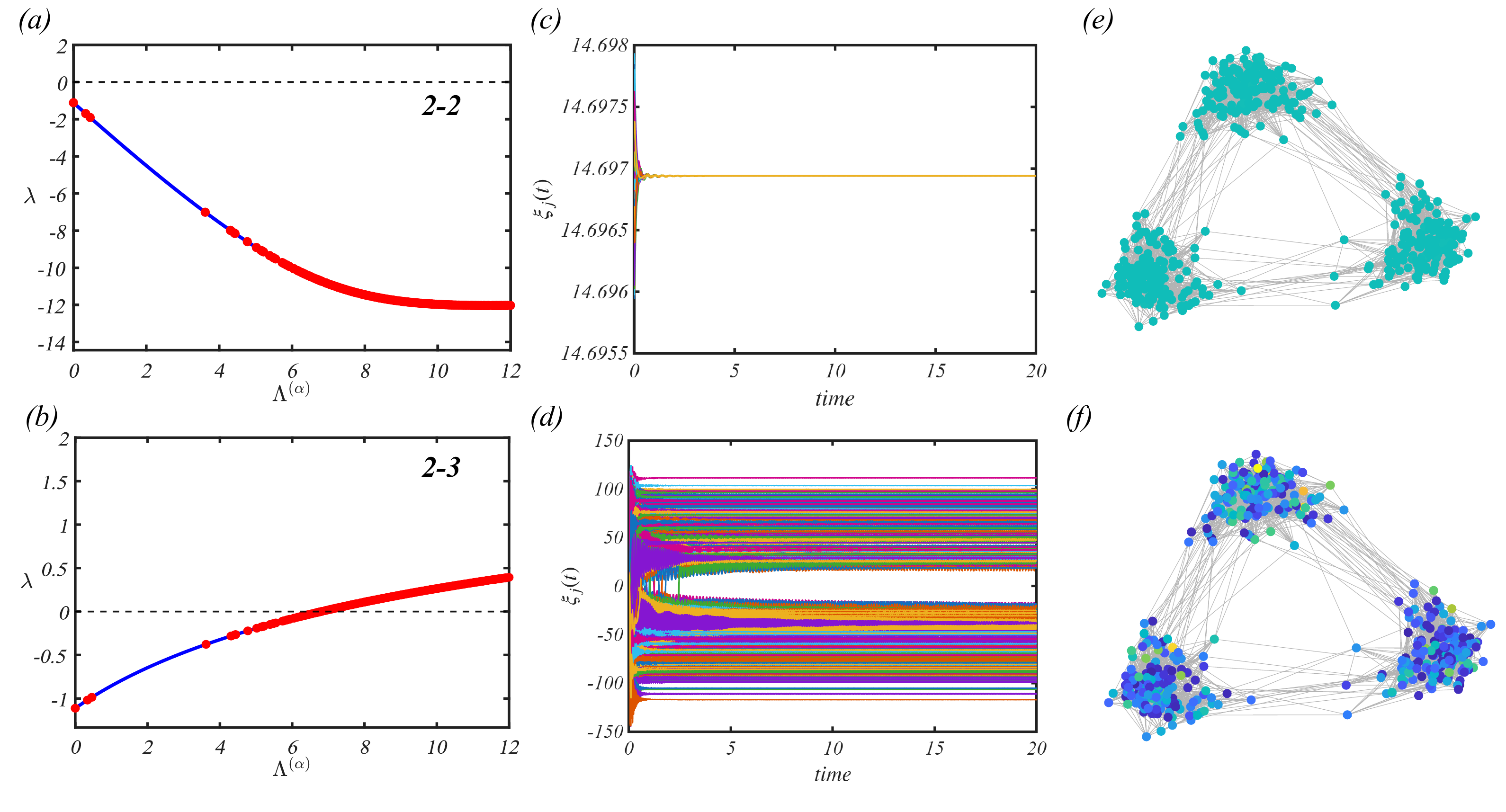}
    \caption{\textbf{Emergence of Turing patterns in a MWN with stochastic block model topology of coupled Lorenz systems.} Panels $(a),(b)$ - The dispersion relations is displayed as a function of network Laplacian eigenvalue $\Lambda^{(\alpha)}$ (reds dots), the blue curve have been obtained by replacing the latter with a continuous variable and it is shown to help the reader to appreciate the nonlinear behavior: $(a)$ the dispersion relation remains negative for any value of $\Lambda^{(\alpha)}$, which prevents the emergence of Turing patterns; $(b)$ negative dispersion relationship up to a critical threshold of $\Lambda^{(\alpha)}$, beyond which it becomes positive enabling the emergence of Turing patterns. Panels $(c),(d)$ - Temporal evolution of $\xi_j(t)$ across nodes: $(c)$ all variables $\xi_j(t)$ converge to the equilibrium point $\xi^*$; $(d)$ emergence of Turing patterns. Panels $(e),(f)$ - Network visualizations with node colors indicating dynamical states, i.e., $\xi_j(t)$ after a sufficiently long period of time: $(e)$ nodes present same color meaning that oscillators assume the same value; $(f)$ nodes present different colors indicating the onset of Turing patterns. The model parameters used to compute the dynamics are $\sigma = 13$, $\rho = 28$, $\beta=8$. The underlying topology is given by a  stochastic block model of Erd\H{o}s–Rényi networks composed by \textcolor{black}{$n = 500$ nodes and with $p_{in} = 0.08$ and $p_{out}=0.001$ and $K=3$ blocks}. The coupling is obtained with $\varepsilon=4$, $E_{22}=1$ (top panels) and $E_{23}=1$ bottom panels).}
    \label{fig:lorenzER2}
\end{figure*}

Another numerical example supporting the analytical findings is shown in Fig.~\ref{fig:lorenzER2}. 
In this case, the underlying structures are generated from a stochastic block model topology  
composed of \textcolor{black}{$n=500$} nodes divided into \textcolor{black}{$K=3$ blocks}. The probability of forming a link 
between two nodes within the same block is set to $p_{\text{in}} = 0.8$, while the 
probability of forming a link between nodes belonging to different blocks is 
$p_{\text{out}} = 0.09$. The model parameters used in this example correspond to $\sigma = 13$, $\beta = 8$, and $\rho = 28$. Top panels of Fig.~\ref{fig:lorenzER2} refer to the linear coupling $E_{22} = 1$ and $\varepsilon=4$. Once again, the dispersion relation is always negative (see panel $(a)$) and Turing patterns cannot emerge since the equilibrium $\vec{x}_j=\vec{x}^*$ for all $j=1,\dots,n$, remains stable in presence of the coupling (see panel $(c)$ and $(e)$). In the bottom panels, Numerical results are presented for $E_{23}=1$ and $\varepsilon=4$. One can observe the existence of eigenvalues $\Lambda^{(\alpha)}\gtrsim 6.08=\gamma^{(crit)}_{23}/\varepsilon$ (where $\gamma^{(crit)}_{23}$ is given by Eq.~\eqref{eq:condTPE23}) for which the dispersion relation is positive (see red dots in panel $(b)$) and thus patterns do emerge (see panels $(d)$ and $(f)$).

\section{Conclusions}
\label{sec:conc}

In this work, we have explored the phenomenon of diffusion-driven instabilities in Matrix-Weighted Networks (MWNs), a recently introduced framework in which the modeling of interactions between node variables relies on matrix weights that encode both interaction strength and directional transformations. Our study has extended the classical theory of Turing pattern formation, originally formulated for systems coupled via scalar-weighted edges, to this more general setting. 

A key notion in our analysis is \textit{coherence}, a structural property of MWNs that ensures the composition of transformation matrices along any oriented cycle equals the identity. As also established in the context of global synchronization~\cite{gallo2025global}, we have shown that coherence is a necessary condition for the emergence of Turing instabilities in MWNs. Indeed, only when the MWN is coherent, one can disentangle the mix of variables created by the transformation weights via the matrix $\mathcal{S}$; the latter allows us to reduce the supra-Laplacian $\mathcal{L}$ to a scalar Laplacian $\bar{\mathcal{L}} = \mathcal{S}\mathcal{L}\mathcal{S}^\top$, whose spectral decomposition can be used to explain the instability analysis. 

Furthermore, we have proposed a novel characterization of coherent MWNs, allowing us to deal with networks of any size and thus to overcome the limitation of the existing literature, where only small hand-made coherent networks have been considered. The proposed algorithm relies on the fact that link transformations can be rewritten as relative rotations between nodes; the method is thus very general and can open the way to many more applications of MWN.

Here, we have considered edge weights defined as rotation matrices, thereby imposing specific structural constraints on the class of admissible dynamical systems, namely, those invariant under the action of such matrices. Notice that the proposed framework can be further generalized to include general orthogonal matrices by imposing additional constraints on the dynamical systems, namely, the invariance with respect to these matrices. 

We have focused on the emergence of diffusion-driven instabilities by considering three different models, namely, the Stuart-Landau (SL) model, an abstract model invariant under $2\pi/k-$ rotations, and the Lorenz system. In the case of the SL model, the analysis was carried out analytically in closed form. Exploiting the rotational symmetry of the dynamics in the $(x,y)$-plane, we have showed that the dispersion relation reduces to a simple spectral condition on the Laplacian eigenvalues, yielding a sharp critical threshold $\Lambda_{\text{crit}} = -\sigma_{\Re}/\mu_{\Re}$: modes with $\Lambda^{(\alpha)} > \Lambda_{\text{crit}}$ become unstable, while the homogeneous mode remains stable. A similar result has been provided for the abstract model invariant under $2\pi/k$ rotations, and it has allowed us to explore the role of discrete rotational symmetries of higher order in the pattern formation mechanism. Finally, for the Lorenz system, we have analytically proved that diagonal couplings can never lead to the emergence of Turing instability, regardless of the network topology. In contrast, we have proved that off-diagonal couplings admit a diffusion-driven instability for appropriate parameter regimes.  In all three settings, all the analytical predictions were confirmed by numerical simulations.

Our finding reveals an interesting interplay between the MWN structure, encoded in its matrix weights and coherence properties, and the dynamical system. Such an interplay is absent in classical scalar-weighted networks and represents a genuinely new feature introduced by the MWN setting. \textcolor{black}{Furthermore, our results have potential applications far beyond the study of pattern formation in networks. The ability to construct matrix-weighted diffusion interactions is particularly relevant to a wide range of real-world systems. For instance, in consensus problems on social networks, matrix-weighted interactions enable the coordination of agents whose states are described by multiple interdependent variables, thereby supporting more robust and reliable collective decision-making~\cite{tian2025matrix}. This framework is also promising for image analysis, where each node may represent an image characterized by several visual features. Matrix-valued interactions can naturally encode transformations, such as rotations, between similar images, leading to more accurate similarity measures and improved image organization, retrieval, and classification, see, e.g.,~\cite{pmlr-v196-barbero22a}.}   

For future investigation, it would be of interest to understand how the breaking of coherence affects the onset of pattern formation and whether approximate coherence can still lead to instabilities. \textcolor{black}{Moreover, another important theoretical question, that deserves to be investigated, concerns the structural robustness of the emerging Turing patterns against small perturbations or violations of the invariance conditions (see Eq.~\eqref{eq:condfO}) required by the proposed framework. 
If such conditions are only approximately satisfied, either due to a slight structural incoherence or due to small non-invariant perturbations in the local dynamics, the standard baseline for a Turing mechanism breaks down: mathematically, a perfectly homogeneous equilibrium (or its rotated counterpart) no longer exists. The system could develop patchy solutions, but they will nor result from a Turing mechanism; instead, they represent forced inhomogeneous steady states sustained by the inherent geometric conflict of the edge rotations. Investigating how the system relaxes toward these perturbed states, and quantifying the structural distance from perfect coherence and invariance using perturbation theory, represents a challenging direction for future work.}

\textcolor{black}{Another future challenge is to extend the definition and the constructive algorithm for guaranteeing coherence to the case of directed (asymmetric) networks: formalizing and verifying the algebraic coherence condition along fully directed cycles introduces intricate mathematical difficulties, leaving the directed matrix-weighted framework as an open problem for future studies.}


\section*{Author contributions}
A.G., W.S., and T.C. contributed equally to the project and the preparation of the manuscript.


\textcolor{black}{
\section*{Code Availability}
The source code used in this study is publicly available at: \url{https://gitlab.unamur.be/codes_turing_patterns_on_mwns}. If using this code, please cite this publication.}

\section*{Competing interests}
The authors declare no competing interests.

\bibliography{mybib.bib}

\clearpage

\onecolumngrid

\appendix

\section{Characterization of coherent MWNs}\label{app:coherence}

This appendix is devoted to provide an in-depth analysis of the characterization of coherent MWNs proposed in Proposition~\ref{prop:coherence}. As in~\cite{tian2025matrix}, we restrict our attention to reciprocal interactions. Accordingly, without loss of generality, we present the discussion in terms of undirected networks.

\begin{lemma}[Coherence on a cycle basis]\label{lemma:cohebasis}
Let $G=(V,E)$ be a MWN graph let $\mathcal B=\{\mathcal C_1,\dots,\mathcal C_\beta\}$ be a cycle basis of $G$.
Assume that for every edge $(i,j)\in E$ the link matrix $\mathbf R_{ij}$ is orthonormal.
If
\begin{align}\label{eq:cohebasis}
\prod_{(i,j)\in \mathcal C_k} \mathbf R_{ij} = \mathbf I_d,
\qquad \forall\, \mathcal C_k \in \mathcal B,
\end{align}
then
\begin{align}
\prod_{(i,j)\in \mathcal C} \mathbf R_{ij} = \mathbf I_d
\end{align}
for any cycle $\mathcal C$ in $G$.
\end{lemma}
\begin{proof} Since $\mathcal B$ is a cycle basis, it follows that any cycle $\mathcal C$ of G can be written as the symmetric difference of a finite number of cycles in $\mathcal B$, i.e., $\mathcal C=\mathcal C_{k_1}\oplus\dots\oplus\mathcal C_{k_m}$. Notice that each edge $(i,j)$ that appears twice in the symmetric difference corresponds to the factor $\mathbf R_{ij}\mathbf R_{ji}$ in the ordered product; hence, since $\mathbf R_{ij}$ is orthonormal, $\mathbf R_{ij}^{-1}=\mathbf R_{ji}$, and the corresponding factor cancels out in the ordered product.

By assumption, Eq.~\eqref{eq:cohebasis} holds, i.e., the product of link matrices along each cycle $\mathcal C_k\in\mathcal B$ equals the identity matrix. Therefore, the product along $\mathcal C$ also equals the identity.
\end{proof}

\begin{figure}[ht!]
    \centering
    \includegraphics[width=0.95\linewidth]{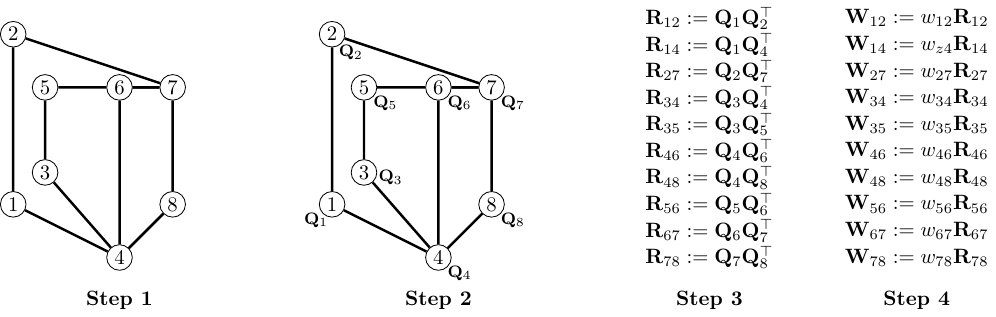}
    \caption{\textbf{Construction of a coherent MWNs.} \textbf{Step 1}: generate underlying topology. \textbf{Step 2}: assign a $d\times d$ orthonormal matrix $\mathbf Q_i\in O(d)$ to each node. \textbf{Step 3}: For each edge, define the transformation matrix as the product of the orthonormal matrix of an endpoint and the transpose of the other one. \textbf{Step 4}: For each edge, define the corresponding weighted matrix as the product between a scalar weight $w_{ij}\in U[w_{\min},w_{\max}]$ and the corresponding transformation matrix.}
    \label{fig:buildcoherent}
\end{figure}

\begin{figure}[ht]
    \centering
    \includegraphics[width=\linewidth]{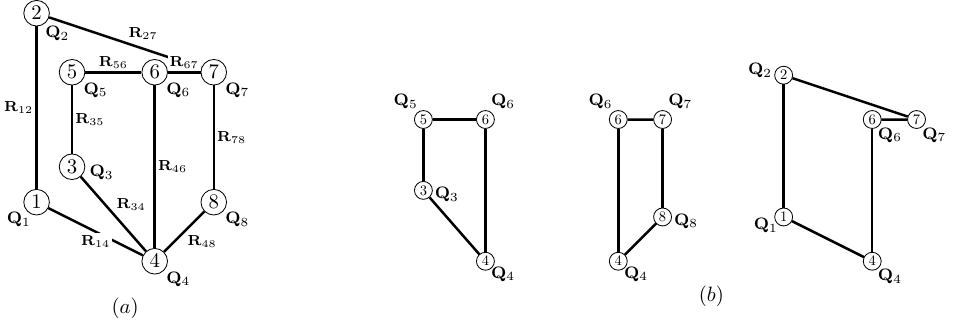}
    \caption{Panel $(a)$ - \textbf{Coherent MWN $G$ defined in Fig.~\ref{fig:buildcoherent}}. $G$ has six distinct cycles, i.e., $\mathcal C_1=(3,5,6,4,3)$, $\mathcal C_2=(4,6,7,8,4)$, $\mathcal C_3=(1,2,7,6,4,1)$, $\mathcal C_4=\{1,2,7,8,4,1\}$, $\mathcal C_5=\{3,4,8,7,6,5,3\}$ and $\mathcal C_6=\{1,2,7,6,5,3,4,1\}$. Panel $(b)$ - \textbf{Cycle basis $\mathcal B$ of $G$}.}
    \label{fig:proofcoherence}
\end{figure}

Let us discuss the illustrative example in Fig.~\ref{fig:proofcoherence}. A cycle basis $\mathcal{B}$ of the graph $G$ in Fig.~\ref{fig:proofcoherence}$(a)$ is given by the three cycles shown separately in fig~\ref{fig:proofcoherence}$(b)$:
\begin{itemize}
    \item $\mathcal{C}_1 = (3, 5, 6, 4, 3)$;
    \item $\mathcal{C}_2 = (4, 6, 7, 8, 4)$;  
    \item $\mathcal C_3=(1,2,7,6,4,1)$.
\end{itemize}
Indeed, notice that we can write the remaining cycles $\mathcal C_4$, $\mathcal C_5$ and $\mathcal C_6$ in terms of such a basis:
\begin{itemize}
    \item $\mathcal{C}_4 = (1,2,7,8,4,1) = \mathcal C_2\oplus\mathcal C_3 = (\mathcal C_2\cup\mathcal C_3)\backslash(\mathcal C_2\cap\mathcal C_3)$;
    \item $\mathcal{C}_5 = (3,4,8,7,6,5,3) = \mathcal C_1\oplus\mathcal C_2 = (\mathcal C_1\cup\mathcal C_2)\backslash(\mathcal C_1\cap\mathcal C_2)$;  
    \item $\mathcal{C}_6 = (1,2,7,6,5,3,4,1) = \mathcal C_1\oplus\mathcal C_3 = (\mathcal C_1\cup\mathcal C_3)\backslash(\mathcal C_1\cap\mathcal C_3)$.
\end{itemize}
To verify the coherence condition, we need to check that the product of link matrices equals identity for each cycle $\mathcal C_1=(3,5,6,4,3)$, $\mathcal C_2=(4,6,7,8,4)$, and $\mathcal C_3=(1,2,7,6,4,1)$:
\begin{align}
\mathcal{C}_1:\quad
\mathbf{R}_{35} \mathbf{R}_{56} \mathbf{R}_{64} \mathbf{R}_{43}
&= (\mathbf{Q}_3^\top \mathbf{Q}_5)
   (\mathbf{Q}_5^\top \mathbf{Q}_6)
   (\mathbf{Q}_6^\top \mathbf{Q}_4)
   (\mathbf{Q}_4^\top \mathbf{Q}_3) \nonumber \\
&= \mathbf{Q}_3^\top
   (\mathbf{Q}_5 \mathbf{Q}_5^\top)
   (\mathbf{Q}_6 \mathbf{Q}_6^\top)
   (\mathbf{Q}_4 \mathbf{Q}_4^\top)
   \mathbf{Q}_3 \nonumber \\
&= \mathbf{Q}_3^\top \mathbf{Q}_3
 = \mathbf{I}_d;
\\[1ex]
\mathcal{C}_2:\quad
\mathbf{R}_{46} \mathbf{R}_{67} \mathbf{R}_{78} \mathbf{R}_{84}
&= (\mathbf{Q}_4^\top \mathbf{Q}_6)
   (\mathbf{Q}_6^\top \mathbf{Q}_7)
   (\mathbf{Q}_7^\top \mathbf{Q}_8)
   (\mathbf{Q}_8^\top \mathbf{Q}_4) \nonumber \\
&= \mathbf{Q}_4^\top
   (\mathbf{Q}_6 \mathbf{Q}_6^\top)
   (\mathbf{Q}_7 \mathbf{Q}_7^\top)
   (\mathbf{Q}_8 \mathbf{Q}_8^\top)
   \mathbf{Q}_4 \nonumber \\
&= \mathbf{Q}_4^\top \mathbf{Q}_4
 = \mathbf{I}_d;
\\[1ex]
\mathcal{C}_3:\quad
\mathbf{R}_{12} \mathbf{R}_{27} \mathbf{R}_{76}
\mathbf{R}_{64} \mathbf{R}_{41}
&= (\mathbf{Q}_1^\top \mathbf{Q}_2)
   (\mathbf{Q}_2^\top \mathbf{Q}_7)
   (\mathbf{Q}_7^\top \mathbf{Q}_6)
   (\mathbf{Q}_6^\top \mathbf{Q}_4)
   (\mathbf{Q}_4^\top \mathbf{Q}_1) \nonumber \\
&= \mathbf{Q}_1^\top
   (\mathbf{Q}_2 \mathbf{Q}_2^\top)
   (\mathbf{Q}_7 \mathbf{Q}_7^\top)
   (\mathbf{Q}_6 \mathbf{Q}_6^\top)
   (\mathbf{Q}_4 \mathbf{Q}_4^\top)
   \mathbf{Q}_1 \nonumber \\
&= \mathbf{Q}_1^\top \mathbf{Q}_1
 = \mathbf{I}_d.
\end{align}
By Lemma~\ref{lemma:cohebasis}, coherence of $\mathcal{B}$ implies coherence of all cycles in the graph. 

Let us now consider the opposite direction of Proposition~\ref{prop:coherence}, i.e., given the coherent MWN $G$ in Fig.~\ref{fig:proofcoherence} to each edge $(i,j)$ to which an orthonormal matrix $\mathbf R_{ij}$ is assigned, we show how to explicitly construct the orthonormal matrices $\mathbf{Q}_1, \ldots, \mathbf{Q}_8 \in O(d)$ such that $\mathbf R_{ij}=\mathbf Q_i^\top\mathbf Q_j$ for any edge $(i,j)$.

First, let us choose a reference node in $V=\{1,\dots,8\}$ and set an initial matrix. Without loss of generality, we select node $k=4$ as the reference one and set $\mathbf Q_4=\mathbf I_d$. 
Second, let us construct the orthonormal matrices to be assigned to the other nodes via paths from the reference node $k=4$. For any $j\in V\backslash\{4\}$, we choose a path from node $4$ to node $j$ and define $\mathbf Q_j:=\mathbf Q_4\mathbf R_{P_{4j}}=\mathbf R_{P_{4j}}$. Explicitly, we then have:
\begin{itemize}
    \item $P_{41} = (4, 1)$: $\mathbf{Q}_1 = \mathbf{R}_{41} = \mathbf{R}_{14}^\top$;    
    \item $P_{42} = (4, 1, 2)$: $\mathbf{Q}_2 = \mathbf{R}_{41} \mathbf{R}_{12} 
        = \mathbf{R}_{14}^\top \mathbf{R}_{12}$;
    \item $P_{43} = (4, 3)$: $\mathbf{Q}_3 = \mathbf{R}_{43} = \mathbf{R}_{34}^\top$;
    \item $P_{45} = (4, 3, 5)$: $\mathbf{Q}_5 = \mathbf{R}_{43} \mathbf{R}_{35} 
        = \mathbf{R}_{34}^\top \mathbf{R}_{35}$
    
    \item $P_{46} = (4, 6)$: $\mathbf{Q}_6 = \mathbf{R}_{46}$;
    \item $P_{47} = (4, 6 , 7)$: $\mathbf{Q}_7 = \mathbf{R}_{46} \mathbf{R}_{67}$;
    \item $P_{48} = (4, 6, 7, 8)$: $\mathbf{Q}_8 = \mathbf{R}_{46} \mathbf{R}_{67} \mathbf{R}_{78}$.
\end{itemize}
Third, to conclude, let us verify that $\mathbf R_{ij}=\mathbf Q_i^\top\mathbf Q_j$ for any edge $(i,j)$. 

For edge $(1,2)$, we have
\begin{align}
    \mathbf Q_1^\top\mathbf Q_2 = (\mathbf R_{14}^\top)^\top\mathbf R_{14}^\top\mathbf R_{12} = \mathbf R_{14}\mathbf R_{14}^\top\mathbf R_{12} = \mathbf R_{12}. 
\end{align}

For edge $(2,7)$, we have:
\begin{align}
    \mathbf Q_2^\top\mathbf Q_7 = (\mathbf{R}_{14}^\top \mathbf{R}_{12})^\top\mathbf{R}_{46} \mathbf{R}_{67} = \mathbf{R}_{12}^\top\mathbf{R}_{14}\mathbf{R}_{46} \mathbf{R}_{67} = \mathbf{R}_{21}\mathbf{R}_{14}\mathbf{R}_{46} \mathbf{R}_{67} = \mathbf R_{27}, 
\end{align}
where the last equality follows by the coherence condition on cycle $\mathcal C_3$. 

For edge $(5,6)$, we have:
\begin{align}
    \mathbf Q_5^\top\mathbf Q_6 = (\mathbf{R}_{34}^\top \mathbf{R}_{35})^\top\mathbf{R}_{46} = \mathbf{R}_{35}^\top\mathbf{R}_{34}\mathbf{R}_{46} = \mathbf R_{56},
\end{align}
where the last equality follows by the coherence condition on cycle $\mathcal C_1$.

For edge $(3,5)$, we have:
\begin{align}
    \mathbf Q_3^\top\mathbf Q_5 = (\mathbf{R}_{43})^\top\mathbf{R}_{34}^\top \mathbf{R}_{35} = \mathbf R_{34}\mathbf{R}_{34}^\top \mathbf{R}_{35} = \mathbf R_{35}.
\end{align}

For edge $(6,7)$, we have:
\begin{align}
    \mathbf Q_6^\top\mathbf Q_7 = (\mathbf{R}_{46})^\top\mathbf{R}_{46} \mathbf{R}_{67} = \mathbf R_{64}\mathbf{R}_{46}\mathbf{R}_{67} = \mathbf{R}_{67}.
\end{align}

For edge $(7,8)$, we have:
\begin{align}
    \mathbf Q_7^\top\mathbf Q_8 = (\mathbf{R}_{46} \mathbf{R}_{67})^\top\mathbf{R}_{46} \mathbf{R}_{67} \mathbf{R}_{78} = \mathbf{R}_{67}^\top\mathbf{R}_{46}^\top\mathbf{R}_{46} \mathbf{R}_{67} \mathbf{R}_{78}=\mathbf{R}_{78}.
\end{align}

Similarly, the proof also holds for the remaining edges $(1,4), (3,4), (4,6)$ and $(4,8)$. Notice that the coherence property ensures that this construction is 
well-defined, i.e., different path choices from node 4 to node $j$ yield 
the same matrix $\mathbf{Q}_j$. For instance, if we consider the path $P_{45}'=(4,6,5)$, we could define $\tilde{\mathbf Q}_5=\mathbf R_{46}\mathbf R_{65}$. Since from the coherence condition on cycle $\mathcal C_1$ we have that $\mathbf R_{46}\mathbf R_{65}\mathbf R_{53}\mathbf R_{34}=\mathbf I_d$, we conclude that
\begin{align}
    \tilde{\mathbf Q}_5 = \mathbf R_{46}\mathbf R_{65} = \mathbf R_{34}^\top\mathbf R_{53}^\top = \mathbf R_{34}^\top\mathbf R_{35} = \mathbf Q_5.
\end{align}

\textcolor{black}{Finally, in the following Algorithm~\ref{alg:coherent_mwn}, we provide a schematic summary of the algorithmic procedure for constructing a coherent MWN.}

\begin{algorithm*}
\DontPrintSemicolon
\KwIn{A baseline scalar, undirected network topology $G = (V, E)$ with $n$ nodes, a set of orthogonal $d\times d$ matrices and two scalar weights $w_{\min},w_{\max}\in\mathbb R$, $w_{\min}<w_{\max}$.}
\KwOut{A coherent matrix-weighted adjacency matrix $\mathbf W$ where $\mathbf W_{ij} = w_{ij}\mathbf R_{ij}$.}
\quad Select an arbitrary root node $k\in V$ and set its coordinate transformation matrix to the identity, $\mathbf Q_k = \mathbf I_d$.\;
\qquad \textbf{for each} node $i \in V\backslash\{k\}$
choose an arbitrary orthogonal transformation matrix $\mathbf Q_i$ from the given set;\;
\qquad \textbf{for each}  $(i, j) \in E$ compute the edge matrix transformation: $\mathbf R_{ij} = \mathbf Q_i^\top \mathbf Q_j$;\;
\qquad define the final matrix weight: $\mathbf W_{ij} = w_{ij}\mathbf R_{ij}$, where $w_{ij}\in U[w_{\min},w_{\max}]$.\;
\caption{\label{alg:coherent_mwn} \textcolor{black}{\textbf{Algorithmic procedure for constructing a coherent MWN.} Algorithmic procedure for constructing a coherent MWN of arbitrary size and with an arbitrary underlying undirected topology.}}
\end{algorithm*}

\textcolor{black}{Let us observe that, because the scalar weights, $w_{ij}$, do not contribute to the coherence or lack thereof, the latter can be arbitrarily chosen.}

\section{Linear stability analysis of the Stuart-Landau model}\label{app:SL}

In this appendix, we present a detailed derivation of the linear stability of a network of Stuart-Landau (SL) oscillators, by including the conditions for diffusion-driven instabilities and the absence of stable limit cycles near the origin.

Let us consider $n$ identical SL oscillators, whose dynamics are governed by

\begin{align}\label{eq:SLapp}
\frac{d}{dt}
\begin{pmatrix}
 x\\y
\end{pmatrix} =&
\begin{pmatrix}
 \sigma_{\Re} & -\sigma_{\Im}\\
 \sigma_{\Im} & \sigma_{\Re}
\end{pmatrix}
\begin{pmatrix}
 x\\y
\end{pmatrix}-  (x^2+y^2)
\begin{pmatrix}
 \beta_{\Re} & -\beta_{\Im}\\
 \beta_{\Im} & \beta_{\Re}
\end{pmatrix}
\begin{pmatrix}
 x\\y
\end{pmatrix},
\end{align}
where $\sigma=\sigma_{\Re}+i\sigma_{\Im}$ and $\beta=\beta_{\Re}+i\beta_{\Im}$ are complex model parameters, and coupled by a diffusive-like nonlinear function
\begin{align}
    \vec{h}(x_\ell,y_\ell):=(x_\ell^2+y_\ell^2)^{\frac{m-1}{2}}
\begin{pmatrix}
 \mu_{\Re} & -\mu_{\Im}\\
 \mu_{\Im} & \mu_{\Re}
\end{pmatrix}
\begin{pmatrix}
 x_\ell\\y_\ell
\end{pmatrix},
\end{align}
with $\mu=\mu_{\Re}+i\mu_{\Im}$ being the complex coupling strength and $m\ge 1$ controlling the nonlinearity of the coupling. The evolution of the $j$--th oscillator is, thus, given by
\begin{align}
\label{eq:SLmatnet}
\frac{d}{dt}
\begin{pmatrix}
 x_j\\y_j
\end{pmatrix} &=
\begin{pmatrix}
 \sigma_{\Re} & -\sigma_{\Im}\\
 \sigma_{\Im} & \sigma_{\Re}
\end{pmatrix}
\begin{pmatrix}
 x_j\\y_j
\end{pmatrix} -  (x_j^2+y_j^2)
\begin{pmatrix}
 \beta_{\Re} & -\beta_{\Im}\\
 \beta_{\Im} & \beta_{\Re}
\end{pmatrix}
\begin{pmatrix}
 x_j\\y_j
\end{pmatrix}-\sum_\ell \mathcal{L}_{j\ell}\left[ (x_\ell^2+y_\ell^2)^{\frac{m-1}{2}}
\begin{pmatrix}
 \mu_{\Re} & -\mu_{\Im}\\
 \mu_{\Im} & \mu_{\Re}
\end{pmatrix}
\begin{pmatrix}
 x_\ell\\y_\ell
\end{pmatrix}\right]\nonumber\\
&=:\vec{f}(x_j,y_j)-\sum_\ell \mathcal{L}_{j\ell}\vec{h}(x_\ell,y_\ell)\, ,
\end{align}
 being $\mathcal{L}$ the supra-Laplace matrix. Through simple algebraic calculations, it is easy to show that such system has an equilibrium point at the origin, i.e., for $x_j=y_j=0$ for every oscillator $j$. 

Let us now restrict to the case $m=1$, so that the nonlinear coupling term results in reality a linear one; let us observe that if $m>1$, the linearization of the coupling term near the origin will vanish. Hance, by linearizing Eq.~\eqref{eq:SLmatnet} about $x_j=y_j=0$ leads to 
\begin{align}
\frac{d}{dt}
    \begin{pmatrix}
 {x}_j\\ {y}_j
\end{pmatrix}=\begin{pmatrix}
 \sigma_{\Re} & -\sigma_{\Im}\\
 \sigma_{\Im} & \sigma_{\Re}
\end{pmatrix}\begin{pmatrix}
 x_j\\y_j
\end{pmatrix}-\sum_{\ell}\mathcal{L}_{\ell j}\begin{pmatrix}
 \mu_{\Re} & -\mu_{\Im}\\
 \mu_{\Im} & \mu_{\Re}
\end{pmatrix}
\begin{pmatrix}
 x_\ell\\y_\ell
\end{pmatrix}.
\end{align}
Disregarding the coupling term, the characteristic polynomial of the local system is
\begin{align}
p(\lambda) = (\sigma_{\Re}-\lambda)^2 + \sigma_{\Im}^2\,,
\end{align}
with eigenvalues
\begin{align}
\lambda = \sigma_{\Re} \pm i \sigma_{\Im}.
\end{align}
Hence, for $\sigma_{\Re}<0$, the origin is a stable solution.

Let now consider $\{\vec{\phi}^{(\alpha)}\}$ to be an orthonormal eigenbasis of $\bar{\mathbf{L}}$, i.e., $\bar{\mathbf{L}}\vec{\phi}^{(\alpha)} = \Lambda^{(\alpha)} \vec{\phi}^{(\alpha)}$, for $\alpha=1,\dots,n$.
Projecting the linear system onto the eigenbasis yields
\begin{align}
    \begin{cases}
        \dfrac{d{x}_j}{dt} &= \sum_\alpha\dfrac{d{\hat{x}}_\alpha}{dt}\vec{\phi}^{(\alpha)}_j = \sum_\alpha(\sigma_{\Re}\hat{x}_\alpha-\sigma_{\Im}\hat{y}_\alpha)\phi^{(\alpha)}_j-\sum_{\ell,\alpha}(\mu_{\Re}\hat{x}_\alpha-\mu_{\Im}\hat{y}_\alpha)L_{j\alpha}\phi^{(\alpha)}_\ell\\
        &= \sum_\alpha(\sigma_{\Re}\hat{x}_\alpha-\sigma_{\Im}\hat{y}_\alpha)\phi^{(\alpha)}_j-\sum_{\alpha}(\mu_{\Re}\hat{x}_\alpha-\mu_{\Im}\hat{y}_\alpha)\Lambda^{(\alpha)}\phi^{(\alpha)}_j, \\
        \dfrac{d{y}_j}{dt} &= \sum_\alpha\dfrac{d{\hat{y}}_\alpha}{dt}\vec{\phi}^{(\alpha)}_j = \sum_\alpha(\sigma_{\Im}\hat{x}_\alpha+\sigma_{\Re}\hat{y}_\alpha)\phi^{(\alpha)}_j-\sum_{\ell,\alpha}(\mu_{\Im}\hat{x}_\alpha+\mu_{\Re}\hat{y}_\alpha)L_{j\alpha}\phi^{(\alpha)}_\ell\\
        &= \sum_\alpha(\sigma_{\Im}\hat{x}_\alpha+\sigma_{\Re}\hat{y}_\alpha)\phi^{(\alpha)}_j-\sum_{\alpha}(\mu_{\Im}\hat{x}_\alpha+\mu_{\Re}\hat{y}_\alpha)\Lambda^{(\alpha)}\phi^{(\alpha)}_j.
    \end{cases}
\end{align}
Then, exploiting the orthogonality of the eigenbasis, we obtain
\begin{align}
    \begin{cases}
        \dfrac{d{\hat{x}}_\alpha}{dt} = \sigma_{\Re}\hat{x}_\alpha-\sigma_{\Im}\hat{y}_\alpha-\Lambda^{(\alpha)}(\mu_{\Re}\hat{x}_\alpha-\mu_{\Im}\hat{y}_\alpha),\\
        \dfrac{d{\hat
        y}_\alpha}{dt} = \sigma_{\Im}\hat{x}_\alpha+\sigma_{\Re}\hat{y}_\alpha - \Lambda^{(\alpha)}(\mu_{\Im}\hat{x}_\alpha+\mu_{\Re}\hat{y}_\alpha),
    \end{cases}
\end{align}
i.e., in a more compact form,
\begin{align}
    \frac{d}{dt}\begin{pmatrix}
        \hat{x}_\alpha \\ \hat{y}_\alpha
\end{pmatrix} = \left[\begin{pmatrix}
 \sigma_{\Re} & -\sigma_{\Im}\\
 \sigma_{\Im} & \sigma_{\Re}
\end{pmatrix}-\Lambda^{(\alpha)}\begin{pmatrix}
 \mu_{\Re} & -\mu_{\Im}\\
 \mu_{\Im} & \mu_{\Re}
\end{pmatrix}\right]\begin{pmatrix}
        \hat{x}_\alpha \\ \hat{y}_\alpha
\end{pmatrix}=:\mathbf{M}(\Lambda^{(\alpha)})\begin{pmatrix}
        \hat{x}_\alpha \\ \hat{y}_\alpha
\end{pmatrix}\, .
\end{align}
The eigenvalues of $\mathbf{M}(\Lambda^{(\alpha)})$, for each mode $\alpha$, read
\begin{align}
    \lambda_\alpha=(\sigma_{\Re}-\Lambda^{(\alpha)}\mu_{\Re}) \pm i|\sigma_{\Im}+\Lambda^{(\alpha)}\mu_{\Im}|.
\end{align}
It follows that the homogeneous mode ($\Lambda^{(1)} = 0$) is stable for $\sigma_{\Re}<0$, while, since $\Lambda^{(\alpha)} > 0$, higher order modes ($\alpha>1$) may become unstable if $\mu_{\Re}$ is \textit{sufficiently} positive, i.e., if
\begin{align}
\mu_{\Re} > - \frac{\sigma_{\Re}}{\Lambda^{(\alpha)}}\, .
\end{align}  
Equivalently, one may conclude that, once the dynamical parameters $\mu_{\Re}$ and $\sigma_{\Re}$ are fixed, the condition for having instability reduces to requiring that the spectrum is sufficiently positive.

Finally, let us now show that the limit cycle does not exist if $\sigma_{\Re}<0$. Notice that, given $r_j^2=x_j^2+y_j^2$, we have
\begin{align}
    2r_j\dot{r}_j&=2x_j(\sigma_{\Re} x_j-\sigma_{\Im} y_j)-2x_jr^2_j(\beta_{\Re} x_j-\beta_{\Im} y_j) + 2y_j(\sigma_{\Im} x_j+\sigma_{\Re} y_j) -2y_jr^2_j(\beta_{\Im} x_j+\beta_{\Re} y_j)\notag\\
    &= 2x_j^2\sigma_{\Re} -2x_jy_j\sigma_{\Im} -2x_j^2r^2_j\beta_{\Re} -2x_jy_jr^2_j\beta_{\Im} +2x_jy_j\sigma_{\Im}+2y_j^2\sigma_{\Re}-2x_jy_jr^2_j\beta_{\Im}-2y^2_jr^2_j\beta_{\Re}\notag\\
    &= 2r^2_j\sigma_{\Re}-2r^4_j\beta_{\Re}=2r_j^2(\sigma_{\Re}-r^2_j\beta_{\Re}),
\end{align}
i.e.,
\begin{align}
    \dot{r}_j &= r_j(\sigma_{\Re}-r^2_j\beta_{\Re}) =: f(r_j),
\end{align}
which is zero if either $r_j=0$ or $r^2_j=\frac{\sigma_{\Re}}{\beta_{\Re}}$. Notice that the first solution $r_j=0$ is stable when $\sigma_{\Re}<0$. On the other hand, since $\sigma_{\Re}<0$, evaluating
\begin{align}
    f'(r_j) = \sigma_{\Re}-3r_j^2\beta_{\Re},
\end{align}
on the cycle radius $r^2_j=\frac{\sigma_{\Re}}{\beta_{\Re}}$, leads to
\begin{align}
    f'\left(\sqrt{\frac{\sigma_{\Re}}{\beta_{\Re}}}\right)=-2\sigma_{\Re}>0.
\end{align}
Thus, for $\sigma_{\Re}<0$, the nontrivial limit cycle turns out to be unstable, and no stable oscillations exist near the origin. Consequently, any instability observed in the network is purely network-driven, not due to the intrinsic limit cycle of the individual oscillators.


\section{Linear stability analysis of the Lorenz MWN}\label{app:lorenz}

In this section, we provide the various conditions under which Turing patterns may emerge in the case of Lorenz oscillators coupled via MWNs. To achieve this goal we separately consider the six cases where $E_{ij}=1$, for a given couple $i,j\in\{1,2,3\}$, $i\neq j$ and all the remaining entries vanish. For each case we determine  the coefficients of the characteristic polynomial of the matrix $\mathbf{M}(\Lambda^{(\alpha)})$ and deduce the conditions for the Turing instability by using the Routh-Hurwitz criterion. To lighten the notation we will use in the following the variable $\gamma=\varepsilon \Lambda^{(\alpha)}$. Let us also recall that we assume
\begin{equation}
\label{eq:stabcondLor}
    \sigma>\beta+1 \text{ and }1<\rho<\rho_H=\frac{\sigma(\sigma+\beta+3)}{\sigma-\beta-1}
\end{equation}
to ensure the stability of the equilibrium $\vec{x}^*=(\sqrt{\beta(\rho-1)},\sqrt{\beta(\rho-1)},\rho-1)^\top$. \textcolor{black}{For a comprehensive overview of the six coupling cases analyzed below, see Table~\ref{tab:lorenz_summary}, which summarizes the respective polynomial coefficients and critical instability thresholds.}

\subsection{Case $\mathbf{E}_{12}=1$}

From the definition of $\mathbf{M}(\gamma)$ given by~\eqref{eq:MlambdaLor} we get 
\begin{align}
    \mathbf{M}(\gamma) =
    \begin{pmatrix}
        -\sigma & \sigma-\gamma & 0 \\
        1 & -1 & -x^* \\
        y^* & x^* & -\beta
    \end{pmatrix}\, ,
\end{align}
and thus the coefficients of the characteristic polynomial are given by
\begin{eqnarray}
a_2 &=& \sigma+\beta+1\, ,\label{eq:a212}\\
a_1 &=&\gamma + \beta\rho+\beta\sigma\, ,\label{eq:a112}\\
a_0 &=& \gamma\beta(2-\rho)+2\beta\sigma(\rho-1)\label{eq:a012}\, .
\end{eqnarray}
We can observe that $a_2>0$ and $a_1>0$ are always positive. A necessary condition for the emergence of Turing patterns is thus $a_0<0$. By using the value of $a_0$ given by~\eqref{eq:a012} we can conclude that $a_0<0$ if
\begin{equation}
\label{eq:TPE12}
    \varepsilon \Lambda^{(\alpha)}=\gamma > 2\sigma \frac{\rho-1}{\rho-2}=\gamma^{(crit)}_{12}\text{ and $\rho>2$}\, ,
\end{equation}
while if $\rho<2$ then $a_0>0$ for all $\gamma>0$, indeed~\eqref{eq:a012} returns
\begin{equation*}
\frac{a_0}{\beta}=\gamma (2-\rho)+2\sigma (\rho-1) >0\, ,
\end{equation*}
because we also have $\rho>1$.

Let us now consider the last condition for the onset of Turing instability, $a_2a_1<a_0$. By using Eqs.~\eqref{eq:a212},~\eqref{eq:a112} and~\eqref{eq:a012} we obtain
\begin{eqnarray*}
  a_2a_1-a_0&= &(\sigma+\beta+1)(\gamma + \beta\rho+\beta\sigma)-\gamma\beta(2-\rho)+2\beta\sigma(\rho-1)=\gamma [\sigma+\beta(\rho-1)+1]+\beta(\rho+\sigma)(\sigma+\beta+1)-2\sigma\beta(\rho-1)=\\
  &=&\gamma [\sigma+\beta(\rho-1)+1]+\beta\rho(\beta+1-\sigma)+\sigma\beta(3+\sigma+\beta)>0\, .
\end{eqnarray*}
Hence $a_2a_1>a_0$. In conclusion Turing patterns can emerge if Eq.~\eqref{eq:TPE12} is satisfied.

\subsection{Case $\mathbf{E}_{13}=1$}

In this case, we get 
\begin{align}
    \mathbf{M}(\gamma) =
    \begin{pmatrix}
        -\sigma & \sigma & -\gamma \\
        1 & -1 & -x^* \\
        y^* & x^* & -\beta
    \end{pmatrix}\, ,
\end{align}
and thus the coefficients of the characteristic polynomial are
\begin{align}
\begin{aligned}
a_2 &= \sigma+\beta+1\, ,\\
a_1 &=\beta\rho+\beta\sigma+\gamma\sqrt{\beta(\rho-1)}\, ,\\
a_0 &= 2\gamma\sqrt{\beta(\rho-1)}+2\beta\sigma(\rho-1)\, .
\end{aligned}
\end{align}
Because $\rho > 1$, we can thus straightforward realize that $a_2>0$, $a_1>0$ and $a_0>0$. 

So Turing patterns emerge if and only if the last Routh-Hurwitz condition, $a_2a_1-a_0<0$ is violated. A direct computation returns
\begin{equation*}
  a_2a_1-a_0= (\sigma+\beta+1)[\beta(\rho+\sigma)+\gamma x^*]-[2\gamma x^*+2\beta\sigma(\rho-1)]=\beta[\rho(-\sigma+\beta+1)+\sigma (\sigma+\beta+3)]+\gamma x^*(\sigma+\beta-1)\, .
\end{equation*}
Let us remember that $\sigma>\beta+1$ and thus $\sigma+\beta-1>2\beta>0$, moreover $\rho (\sigma-\beta-1)<\sigma(\sigma+\beta+3)$, we can thus conclude that
\begin{equation*}
  a_2a_1-a_0=\beta[\rho(-\sigma+\beta+1)+\sigma (\sigma+\beta+3)]+\gamma x^*(\sigma+\beta-1)>0\, .
\end{equation*}
In conclusion, in the case $E_{13}=1$, Turing patterns never emerge.

\subsection{Case $\mathbf{E}_{21}=1$}

The matrix $\mathbf{M}(\gamma)$ is now given by
\begin{align}
    \mathbf{M(\gamma}) =
    \begin{pmatrix}
        -\sigma & \sigma & 0 \\
        1-\gamma & -1 & -x^* \\
        y^* & x^* & -\beta
    \end{pmatrix}\, .
\end{align}
The coefficients of the characteristic polynomial are
\begin{eqnarray}
a_2 &=& \sigma+\beta+1\, ,\label{eq:a221}\\
a_1 &=& \beta(\rho+\sigma)+\gamma\sigma\, ,\label{eq:a121}\\
a_0&=&\beta[2\rho\sigma+\gamma\sigma-2\sigma]\, .
\label{eq:a021}
\end{eqnarray}
Clearly $a_2>0$ and $a_1>0$, moreover
\begin{equation*}
    a_0=\beta[2\rho\sigma+\gamma\sigma-2\sigma]=\beta[2\sigma(\rho-1)+\gamma\sigma-]>0\, ,
\end{equation*}
because $\rho>1$.
It remains to check the last condition, i.e.,
\begin{eqnarray*}
  &\, &a_2a_1-a_0= \\
  &=&(\sigma+\beta+1)[\beta(\rho+\sigma)+\gamma\sigma]-\beta[2\rho\sigma+\gamma\sigma-2\sigma]=\beta[\sigma\rho+\sigma^2+\beta\rho+\beta\sigma+\rho+\sigma]+\gamma\sigma(\sigma+\beta+1)-\beta[2\sigma(\rho-1)+\gamma\sigma]=\\
  &=&\beta[-\rho(\sigma-\beta-1)+\sigma(\sigma+\beta+3)]+\gamma\sigma (\sigma+1)\, ,
\end{eqnarray*}
by recalling again $\rho (\sigma-\beta-1)<\sigma(\sigma+\beta+3)$, we can conclude that $a_2a_1-a_0>0$, hence not Turing patterns can develop in this case.

\subsection{Case $\mathbf{E}_{23}=1$}

In this case, the matrix $\mathbf{M}(\gamma)$ is
\begin{align}
    \mathbf{M}(\gamma) =
    \begin{pmatrix}
        -\sigma & \sigma & 0 \\
        1 & -1 & -x^*-\gamma \\
        y^* & x^* & -\beta
    \end{pmatrix}\, ,
\end{align}
and the coefficients of the characteristic polynomial are
\begin{align}
\begin{aligned}
a_2 &= \sigma+\beta+1\, ,\\
a_1 &=\beta(\rho+\sigma)+\gamma\sqrt{\beta(\rho-1)}\, ,\\
a_0 &= 2\beta\sigma(\rho-1)+2\gamma\sigma\sqrt{\beta(\rho-1)}\, .
\end{aligned}
\end{align}
The three coefficients $a_2$, $a_1$ and $a_0$ are strictly positive, and therefore the emergence of patterns depends solely on the fourth Routh-Hurwitz condition based on the sign of $a_2a_1-a_0$. 

We then have
\begin{eqnarray*}
  &\, &a_2a_1-a_0= \\
  &=&(\sigma+\beta+1)[\beta(\rho+\sigma)+\gamma x^*]-2\sigma[\beta(\rho-1)+\gamma x^*]=\gamma x^*(-\sigma+\beta+1)+\beta[(\rho+\sigma)(\sigma+\beta+1)-2\sigma(\rho-1)]=\\
  &=&-\gamma x^*(\sigma-\beta-1)+\beta[-\rho(\sigma-\beta-1)+\sigma(\sigma+\beta+3)]\, ,
\end{eqnarray*}
observe that $\sigma-\beta-1>0$ and $-\rho(\sigma-\beta-1)+\sigma(\sigma+\beta+3)>0$, hence $a_2a_1-a_0<0$ if
\begin{equation}
    \label{eq:condTPE23}
    \varepsilon \Lambda^{(\alpha)}=\gamma > \beta\frac{\sigma(\sigma+\beta+3)-\rho(\sigma-\beta-1)}{x^*(\sigma-\beta-1)}=\gamma^{(crit)}_{23}\, ,
\end{equation}

\subsection{Case $\mathbf{E}_{31}=1$}

This case has already been considered in the main text; the aim of this section is to prove the negativity of $\gamma_2$ defined in Eq.~\eqref{eq:cond2Lor}, which is hereby recalled for the sake of simplicity.
\begin{equation*}
\gamma_2=2\sqrt{\beta(\rho-1)}-\frac{\beta(\rho+\sigma)(\beta+\sigma+1)}{\sigma\sqrt{\beta(\rho-1)}}\, .
\end{equation*}
Let us rewrite $\gamma_2$ as follows
\begin{eqnarray*}
\gamma_2&=&\frac{\beta}{\sigma\sqrt{\beta(\rho-1)}} \left[2\sigma (\rho-1)-(\rho+\sigma)(\beta+\sigma+1)\right]=\frac{\beta}{\sigma\sqrt{\beta(\rho-1)}} \left[2\rho\sigma-2\sigma-\sigma(\beta+\sigma+1)-\rho(\beta+1)-\rho\sigma\right]=\\
&=&
\frac{\beta}{\sigma\sqrt{\beta(\rho-1)}} \left[\rho(\sigma-\beta-1)-\sigma(\beta+\sigma+3)\right]\, .
\end{eqnarray*}
Let us recall that by assumption
\begin{equation*}
    \rho<\frac{\sigma(\sigma+\beta+3)}{\sigma-\beta-1}\, ,
\end{equation*}
from which it follows that $\gamma_2<0$.

\subsection{Case $\mathbf{E}_{32}=1$}

In this case we have 
\begin{align}
    \mathbf{M}(\gamma) =
    \begin{pmatrix}
        -\sigma & \sigma & 0 \\
        1 & -1 & -x^* \\
        y^* & x^*-\gamma & -\beta
    \end{pmatrix}\, ,
\end{align}
and the coefficients of the characteristic polynomial are given by
\begin{align}
\begin{aligned}
a_2 &= \sigma+\beta+1\, ,\\
a_1 &= \beta(\rho+\sigma)-\gamma\sqrt{\beta(\rho-1)}\, ,\\
a_0 &= 2\beta\sigma(\rho-1)-\gamma\sigma\sqrt{\beta(\rho-1)}\, .
\end{aligned}
\end{align}
One can easily realize that $a_0<0$ if
\begin{equation}
    \label{eq:condTP320}
    \varepsilon \Lambda^{(\alpha)}=\gamma >2\sqrt{\beta(\rho-1)}=x^*\, ,
\end{equation}
and $a_1<0$ if
\begin{equation}
    \label{eq:condTP321}
    \varepsilon \Lambda^{(\alpha)}=\gamma >\frac{\beta(\rho+\sigma)}{\sqrt{\beta(\rho-1)}}=\gamma^{(crit)}_{32}\, .
\end{equation}
Moreover we have $a_2a_1-a_0<0$ if and only if
\begin{equation*}
\varepsilon \Lambda^{(\alpha)}=\gamma >\frac{
\beta(\rho+\sigma)(\sigma+\beta+1)+2\sigma \beta(\rho-1)}{(\beta+1)\sqrt{\beta(\rho-1)}}
=\tilde{\gamma}^{(crit)}_{32}\, .
\end{equation*}
In conclusion Turing patterns emerge if 
\begin{equation}
    \label{eq:condTP32}
    \varepsilon\Lambda^{(\alpha)} >\min\Big\{x^*,\gamma^{(crit)}_{32},\tilde{\gamma}^{(crit)}_{32}\Big\}\, .
\end{equation}

\begin{table*}[ht!]
\centering
\begin{ruledtabular}
\begin{tabular}{ccccc}
\textbf{Case} & $a_2$ & $a_1$ & $a_0$ & \textbf{Instability Condition} \\ \hline
$E_{12} = 1$ & $\sigma+\beta+1$ & $\gamma+\beta\rho+\beta\sigma$ & $\gamma\beta(2-\rho)+2\beta\sigma(\rho-1)$ & $\rho > 2$ and $\gamma > \gamma_{12}^{(crit)} = 2\sigma\frac{\rho-1}{\rho-2}$ \\
$E_{13} = 1$ & $\sigma+\beta+1$ & $\beta\rho+\beta\sigma+\gamma\sqrt{\beta(\rho-1)}$ & $2\gamma\sqrt{\beta(\rho-1)}+2\beta\sigma(\rho-1)$ & Turing patterns never emerge ($a_2a_1 - a_0 > 0$) \\
$E_{21} = 1$ & $\sigma+\beta+1$ & $\beta\rho+\beta\sigma+\gamma\sigma$ & $\beta[2\sigma(\rho-1)+\gamma\sigma]$ & Turing patterns never emerge ($a_2a_1 - a_0 > 0$) \\
$E_{23} = 1$ & $\sigma+\beta+1$ & $\beta(\rho+\sigma)+\gamma\sqrt{\beta(\rho-1)}$ & $2\beta\sigma(\rho-1)+2\gamma\sigma\sqrt{\beta(\rho-1)}$ & $\gamma > \gamma_{23}^{(crit)} = \beta\frac{\sigma(\sigma+\beta+3)-\rho(\sigma-\beta-1)}{\sqrt{\beta(\rho-1)}(\sigma-\beta-1)}$ \\
$E_{31} = 1$ & $\sigma+\beta+1$ & $\beta(\rho+\sigma)$ & $2\sigma\beta(\rho-1)-\gamma\sigma\sqrt{\beta(\rho-1)}$ & $\gamma > \gamma_{31}^{(crit)} = 2\sqrt{\beta(\rho-1)}$ \\
$E_{32} = 1$ & $\sigma+\beta+1$ & $\beta(\rho+\sigma)-\gamma\sqrt{\beta(\rho-1)}$ & $2\beta\sigma(\rho-1)-\gamma\sigma\sqrt{\beta(\rho-1)}$ & $\gamma > \min\{x^*, \gamma_{32}^{(crit)}, \tilde{\gamma}_{32}^{(crit)}\}$ \\
\end{tabular}
\end{ruledtabular}
\caption{\textcolor{black}{\textbf{Summary of the Turing instability conditions for six off-diagonal coupling cases of the Lorenz system.} Summary of the Routh-Hurwitz polynomial coefficients and Turing instability conditions for the six off-diagonal coupling cases of the Lorenz system, with $\gamma = \epsilon\Lambda^{(\alpha)}$ and $\sigma > \beta + 1$.}}
\label{tab:lorenz_summary}
\end{table*}

\textcolor{black}{
\section{Proof of the commutation and invariance conditions}
\label{app:proof_eq35}
We briefly recall that the dynamics of each of the $n$ identical Stuart-Landau oscillators anchored to the nodes of a MWN is given by
\begin{align}
\label{eq:SLnonlinApp}
\frac{d}{dt}
\begin{pmatrix}
 x_j\\y_j
\end{pmatrix} &=
\begin{pmatrix}
 \sigma_{\Re} & -\sigma_{\Im}\\
 \sigma_{\Im} & \sigma_{\Re}
\end{pmatrix}
\begin{pmatrix}
 x_j\\y_j
\end{pmatrix} -  (x_j^2+y_j^2)
\begin{pmatrix}
 \beta_{\Re} & -\beta_{\Im}\\
 \beta_{\Im} & \beta_{\Re}
\end{pmatrix}
\begin{pmatrix}
 x_j\\y_j
\end{pmatrix}-\sum_\ell \mathcal{L}_{j\ell}\left[ (x_\ell^2+y_\ell^2)^{\frac{m-1}{2}}
\begin{pmatrix}
 \mu_{\Re} & -\mu_{\Im}\\
 \mu_{\Im} & \mu_{\Re}
\end{pmatrix}
\begin{pmatrix}
 x_\ell\\y_\ell
\end{pmatrix}\right]\nonumber\\
&=:\vec{f}(x_j,y_j)-\sum_\ell \mathcal{L}_{j\ell}\vec{h}(x_\ell,y_\ell)\, ,
\end{align}
where $\vec f(x_j,y_j)$ is the nonlinear function defined by the above equation and $\vec{h}(x_\ell,y_\ell):=(x_\ell^2+y_\ell^2)^{\frac{m-1}{2}}
\begin{pmatrix}
 \mu_{\Re} & -\mu_{\Im}\\
 \mu_{\Im} & \mu_{\Re}
\end{pmatrix}
\begin{pmatrix}
 x_\ell\\y_\ell
\end{pmatrix}$ defines the coupling function with complex coupling strength $\mu=\mu_{\Re}+i\mu_{\Im}$, and $\mathcal L$ is the supra-Laplace matrix of the MWN. In this Appendix, we show that the invariance conditions
\begin{align}
   \vec f(\mathbf{R}\vec{x}) = \mathbf R\vec f(\vec x)\quad\text{and}\quad \vec h(\mathbf R\vec x) = \mathbf R\vec h(\vec x),
\end{align}
hold for any general $2 \times 2$ orthogonal matrix $\mathbf R\in SO(2)$, for which the following holds true $\det\mathbf R = 1$.
}

\textcolor{black}{
Note that every real $\mathbf R\in SO(2)$ with $\det\mathbf R=1$ can be uniquely parameterized by a single rotation angle $\theta$. Therefore, we restrict the proof to the canonical trigonometric form $\mathbf R=\begin{pmatrix}
    \cos\theta & -\sin\theta \\
    \sin\theta & \cos\theta
\end{pmatrix}$, without loss of generality. 
The invariance condition, then, follows by simple algebraic computations:
\begin{align}
\begin{pmatrix}
    \sigma_{\Re} & -\sigma_{\Im} \\
    \sigma_{\Im} & \sigma_{\Re}
\end{pmatrix}
\begin{pmatrix}
    \cos\theta & -\sin\theta \\
    \sin\theta & \cos\theta
\end{pmatrix}\begin{pmatrix}
    x_1 \\ x_2
\end{pmatrix} &= \begin{pmatrix}
    \sigma_{\Re} & -\sigma_{\Im} \\
    \sigma_{\Im} & \sigma_{\Re}
\end{pmatrix}
\begin{pmatrix}
    x_1\cos\theta-x_2\sin\theta \\
    x_1\sin\theta+x_2\cos\theta
\end{pmatrix} \notag\\
&=
\begin{pmatrix}
    x_1\sigma_{\Re}\cos\theta-x_2\sigma_{\Re}\sin\theta-x_1\sigma_{\Im}\sin\theta-x_2\sigma_{\Im}\cos\theta \\
    x_1\sigma_{\Im}\cos\theta-x_2\sigma_{\Im}\sin\theta+x_1\sigma_{\Re}\sin\theta+x_2\sigma_{\Re}\cos\theta
\end{pmatrix}\notag\\
&=\begin{pmatrix}
    x_1(\sigma_{\Re}\cos\theta-\sigma_{\Im}\sin\theta)-x_2(\sigma_{\Re}\sin\theta+\sigma_{\Im}\cos\theta) \\   x_1(\sigma_{\Im}\cos\theta+\sigma_{\Re}\sin\theta)+x_2(\sigma_{\Re}\cos\theta-\sigma_{\Im}\sin\theta)
\end{pmatrix}
\end{align}
and,
\begin{align}
\begin{pmatrix}
    \cos\theta & -\sin\theta \\
    \sin\theta & \cos\theta
\end{pmatrix}
\begin{pmatrix}
    \sigma_{\Re} & -\sigma_{\Im} \\
    \sigma_{\Im} & \sigma_{\Re}
\end{pmatrix}\begin{pmatrix}
    x_1 \\ x_2
\end{pmatrix} &=
\begin{pmatrix}
    \cos\theta & -\sin\theta \\
    \sin\theta & \cos\theta
\end{pmatrix}\begin{pmatrix}
    x_1\sigma_{\Re}-x_2\sigma_{\Im} \\
    x_1\sigma_{\Im}+x_2\sigma_{\Re}
\end{pmatrix} \notag\\
&=
\begin{pmatrix}
    x_1\sigma_{\Re}\cos\theta-x_2\sigma_{\Im}\cos\theta-x_1\sigma_{\Im}\sin\theta-x_2\sigma_{\Re}\sin\theta \\
    x_1\sigma_{\Re}\sin\theta-x_2\sigma_{\Im}\sin\theta+x_1\sigma_{\Im}\cos\theta+x_2\sigma_{\Re}\cos\theta
\end{pmatrix}\notag\\
&=\begin{pmatrix}
    x_1(\sigma_{\Re}\cos\theta-\sigma_{\Im}\sin\theta)-x_2(\sigma_{\Re}\sin\theta+\sigma_{\Im}\cos\theta) \\   x_1(\sigma_{\Im}\cos\theta+\sigma_{\Re}\sin\theta)+x_2(\sigma_{\Re}\cos\theta-\sigma_{\Im}\sin\theta)
\end{pmatrix}.
\end{align} 
Similarly, the computations above hold true for $\beta$ and $\mu$.
}

\end{document}